\documentclass[11pt]{article}
\pdfoutput=1
\usepackage[utf8]{inputenc}
\usepackage[OT1]{fontenc}
\usepackage[usenames]{color}

\usepackage[protrusion=true,expansion=true,final,babel]{microtype}

\usepackage{fullpage}

\usepackage[colorlinks,linkcolor=red,anchorcolor=blue,citecolor=blue]{hyperref}

\usepackage{natbib}

\usepackage[subcaption]{styles/smile}
\usepackage{styles/cme-math}
\usepackage{cancel}

\crefname{lemma}{lemma}{lemmas}
\Crefname{lemma}{Lemma}{Lemmas}

\newcommand{\ignore}[1]{}

\newif\ifamazoninternal
\amazoninternaltrue
\amazoninternalfalse
\newcommand{\amazon}[2]{\ifamazoninternal #1\else  #2\fi}

\begin{document}

\title{Neural Coordination and Capacity Control for Inventory Management}

\author{
Carson Eisenach\thanks{SCOT, Amazon}
	\and 
Udaya Ghai\footnotemark[1]
	\and 
Dhruv Madeka\thanks{Google. Work done while at Amazon.}
	\and
Kari Torkkola\footnotemark[1]
	\and
Dean Foster\footnotemark[1]
	\and
Sham Kakade\footnotemark[1]~\thanks{Harvard University, Cambridge, MA.}
}

\maketitle

\begin{abstract}
	This paper addresses the capacitated periodic review inventory control problem, focusing on a retailer managing multiple products with limited shared resources, such as storage or inbound labor at a facility. Specifically, this paper is motivated by the questions of (1) what does it mean to backtest a capacity control mechanism, (2) can we devise and backtest a capacity control mechanism that is compatible with recent advances in deep reinforcement learning for inventory management? First, because we only have a single historic sample path of Amazon's capacity limits, we propose a method that samples from a distribution of possible constraint paths covering a space of real-world scenarios. This novel approach allows for more robust and realistic testing of inventory management strategies. Second, we extend the exo-IDP (Exogenous Decision Process) formulation of \citet{madeka2022deep} to capacitated periodic review inventory control problems and show that certain capacitated control problems are no harder than supervised learning. Third, we introduce a `neural coordinator', designed to produce forecasts of capacity prices, guiding the system to adhere to target constraints in place of a traditional model predictive controller. Finally, we apply a modified DirectBackprop algorithm for learning a deep RL buying policy and a training the neural coordinator. Our methodology is evaluated through large-scale backtests, demonstrating RL buying policies with a neural coordinator outperforms classic baselines both in terms of cumulative discounted reward and capacity adherence (we see improvements of up to 50\% in some cases).
\end{abstract}


\section{Introduction}
\label{sec:intro}

In modern inventory control systems, managing capacity resources that are shared across hundreds or thousands of products is a key problem. We are interested in the setting where a large retailer manages a supply chain for multiple products and has limited resources (such as storage) that are shared amongst all the products that retailer stocks. This is known as the {\it capacitated} periodic review inventory control problem, and in this work we consider the setting where a retailer with a single facility stocks multiple products and seeks to maximize revenue subject to volumetric (storage) and/or flow constraints (i.e. into or out of a warehouse) at that facility.\footnote{We consider the single facility problem for reasons of clarity, but our methodology could be extended to a multi-facility network\amazon{ (see \citet{jia2023drim})}{}.}

Many variants of the inventory control problem have been studied extensively in the operations research (OR) literature \citep{scarf1959opt,porteus2002foundations}. More recently, several works have applied deep reinforcement learning to the unconstrained problem \citep{madeka2022deep,andaz2023learning\amazon{,andaz2022rlqot}{}} and have shown how to use historic data from an actual supply chain to construct a simulator for policy learning and evaluation. On the theoretical side, this line of work has shown that for certain inventory control problems -- namely, those that can be cast in the exogenous decision process framework \citep{sinclair2023hindsight,madeka2022deep} -- one can reduce the reinforcement learning problem to supervised learning.

\begin{figure}[h]
	\centering
	\amazon{
		\input{internal_details/constraints_figure.tex}
	}{
		\includegraphics[width=0.45\textwidth]{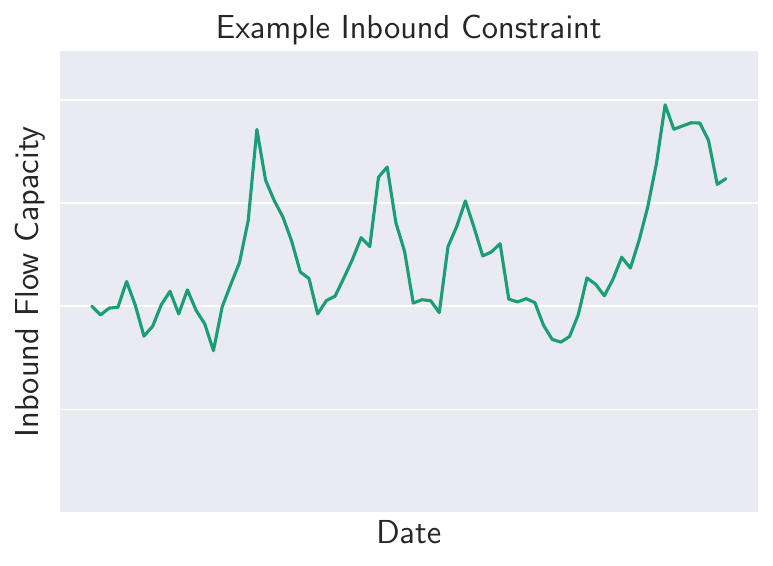} ~~~~~
		\includegraphics[width=0.45\textwidth]{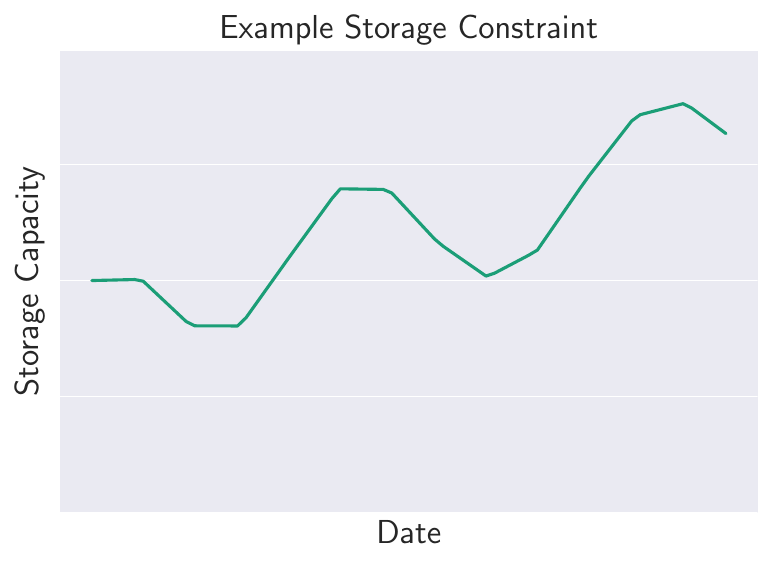}
		\caption{Stylized examples of inbound flow and storage constraints at a retailer.}
	}%
	\label{fig:sort-final-constraints}
\end{figure}

Because it is not tractable to solve the multi-product problem directly using dynamic programming, the canonical method for handling network constraints \citep{maggiar2022multi} is to solve a Lagrangian relaxation, thereby recovering separability between the product-level optimization problems. In practice, model predictive control is used to make replenishment decisions and set shadow prices on the shared resource\amazon{ \citep{ozkok2013acc}}{}.

Recent work showed that Deep RL based policies can improve profitability over sophisticated base stock policies in a series of large-scale real-world A/B tests at Amazon \citep{madeka2022deep,andaz2023learning\amazon{,andaz2022rlqot}{}}. In this paper we extend these approaches to handle {\it capacity constraints}. Motivated by similar questions of {\it learnability} and {\it backtestability}, we address the following gaps in the existing literature.

First, for any inventory management system, what does it mean to backtest a capacity control mechanism? Because capacity constraints are at the facility level (rather than the product-level), we only have a single sample path of the constraints Amazon's supply chain\footnote{Unlike backtesting the product-level rewards where we have $N$ replicates}.  To this end, we propose sampling from a distribution over many possible constraint paths that have ``similar'' structure to real world constraint paths. The key idea is that we can cover the space of all {\it possible} paths the agent might have seen. Because the coupling between the product level decision processes is weak, we are still able to obtain backtesting guarantees.

Second, while model predictive control could be used with an RL inventory control policy, ti do so would require forward simulating every feature the policy takes as input. For example, while traditional inventory control systems might require forecasts of demand, a Deep RL agent would require simulation of every feature it uses as input such as customer arrivals, costs and prices. Thus to use model predictive control, one must choose between restricting the feature set available to the RL agent or modelling the joint distribution of many complex processes. In this work, we propose a novel approach to solving this problem where we {\it forecast} the distribution of dual costs necessary to constrain the policy. We term this the {\it neural coordinator} which, given a target trajectory, produces forecasts of capacity prices that will constrain to the target. 
The neural coordinator approach is not limited to deep RL policies and can be used with traditional inventory management systems as well. As we show in \Cref{sec:experiment}, the neural coordinator produces price trajectories (for a fixed target date) that are martingale, unlike a model predictive controller.

Finally, we extend the exo-IDP formulation in \citet{madeka2022deep} to a class of capacitated periodic review inventory control problems and establish learnability results. This means we can {\it backtest} the capacitated inventory management system before using it in the real-world. We propose a modified DirectBackprop algorithm to learn a buying policy as well as a training scheme for the neural coordinator. We evaluate our proposed methodology through large-scale backtests and show that the neural coordinator with RL buying policy outperforms classic baselines in terms of both constraint adherence and cumulative discounted reward.

\section{Background and Related Work}
\label{sec:background}

\subsubsection*{Model Predictive Control}

Model predictive control consists of using a model to forward simulate a system to optimize control inputs and satisfy any constraints. At each time step, one re-plans based on updated information that has become available in order to select the next control input \citep{kwon1983stabilizing,garcia1989mpc,camacho2004model}. This is sometimes called {\it receding horizon control} because at each time step, one replans with a time horizon one step further in the future. 

\subsubsection*{Inventory Control}

Inventory control systems have been studied extensively in the literature under a variety of conditions (see \citet{porteus2002foundations} for a comprehensive overview). The simplest form is the newsvendor, which solves a myopic problem \citep{arrow1958studies}. Many extensions exist \citep{nahmias1979simple, arrow1958studies}, and the optimal policy in many variants takes the form of a {\it base stock policy} which consists of setting a target inventory level and then placing orders up to that level. \nocite{zipkin2008old} More recently, several works have applied Deep Reinforcement Learning to the inventory management problem \citep{madeka2022deep,\amazon{eisenach2022capacity,andaz2022rlqot,}{}andaz2023learning,alvo2023neural,mousa2023analysis,zhao2023policy,thomas2023towards,parmas2023model,gijsbrechts2022can, qi2023practical}

In the literature, several works have considered constrained inventory management. Typical settings include a production facility where a machine must be shared amongst the production of multiple products \citep{bretthauer1994model} or a retailer which has limited storage space and $N$ items \citep{maloney1993constrained, rosenblatt1990single, ziegler1982solving, ventura1988note, rosenblatt1981multi}. Constraints can also be across facilities, e.g. where a fixed quantity of goods must be split across multiple stores \citep{caro2010inventory, alvo2023neural}. The literature on assortment optimization considers a related problem where one has a constraint on the number of products that can be offered \citep{lo2021omnichannel}.

\nocite{xie2023vcinv}

\subsubsection*{Coordination Mechanisms}

A common formulation is to take a Lagrangian relaxation \citep{boyd2004convex} and couple across products via a resource cost \citep{maloney1993constrained}. In some cases, a dual ascent is performed (e.g. with ADMM) \citep{ziegler1982solving}, and in other cases, closed form solutions exist for (or heuristics are used to approximate) the optimal Lagrange multiplier \citep{rosenblatt1981multi, ventura1988note, rosenblatt1990single}. Recently, \citet{maggiar2024cpp} proposed the {\it Consensus Planning Protocol} (CPP) which considers multiple agents each optimizing their own utility. These agents have joint constraints (e.g. a shared resource) and the overall objective is to optimize the sum of utilities. This is closely related to a distributed ADMM procedure \citep{boyd2004convex,boyd2011foundations}. Our $N+1$-agent formulation of capacitated inventory management is a special case of the setting considered in \citet{maggiar2024cpp}, and we note that CPP could also be used in conjunction with our proposed neural coordinator.

\subsubsection*{Reinforcement Learning}

Reinforcement learning has been applied to sequential decision-making problems including games and simulated physics models \citep{silver2016mastering,szepesvari2010algorithms,mnih2013playing, sutton2020reinforcement,schulman2017proximal,mnih2016asynchronous}. Although in general, one can require exponentially many samples to learn a control policy, recent work has considered a class of decision problems where sample-efficient backtesting is possible just as in supervised learning \citep{madeka2022deep, sinclair2023hindsight}. These are called {\it exogenous interactive decision processes} wherein the state consists of a stochastic exogenous process (independent of the control) and an endogenous component that is governed by a known transition function $f$ of both the previous endogenous state and the exogenous process. Our proposed NCC is also related to imitation learning and has to handle similar issues with non-i.i.d. data \citep{ross2011reduction}.

We also build off recent work in the time-series forecasting literature \citep{oord2016wavenet,wen2017mqcnn,eisenach2020mqt}, and use similar architectures to construct our neural coordinator and for the buying agent's policy network.

\amazon{
	\input{internal_details/amazon_related.tex}
}{}

\subsubsection*{Mathematical notation}
Denote by $\mathbb{R}$, $\mathbb{R}_{\geq 0}$, $\mathbb{Z}$, and $\mathbb{Z}_{\geq 0}$ the set of reals, non-negative reals, integers, and non-negative integers, respectively. We let $(\cdot)_+$ refer to the classical positive part operator i.e. $(\cdot)^{+} = \max(\cdot, 0)$. Let $[\;\cdot\;]$ refer to the set of positive integers up to the argument, i.e. $[\;\cdot\;] = \{x \in \mathbb{Z} \, | \, 1 \leq x \leq \cdot \;\}$.  We use $\mathbb{E}^{\mathbb{P}}$ to denote an expectation operator of a random variable with respect to some probability measure $\mathbb{P}$. Let $||X,Y||_{TV}$ denote the total variation distance between two probability measures $X$ and $Y$.

\section{Constrained Sequential Decision Problems}
\label{sec:ncc-generic}
In this section we consider a general, constrained sequential decision problem similar to our target application in \Cref{sec:idp}; note that the notation defined in this section will not be used in the remainder of the paper.

The class of problem we are interested in solving are constrained Markov decision processes where for each time $t$, we have a set of $M$ constraints $g^m_{t}(s_t) \leq K_{t,m}$. Let $\Kb$ denote the sequence of constraints and similarly let $\Kb_{t,:} \in \RR^{M}$ denote the vector of constraints at time $t$. The constrained Markov decision process is described by the tuple $\cM := (\cS,\cX,\Kb,\PP,R,\gamma,s_0)$, where $\cS$ is the state space, $\cX$ the action space, $\PP$ the probability transition kernel,  $R : \cS\times\cS\times\cX \rightarrow [0,1]$ the reward function, discount factor $\gamma \in [0,1)$, and initial state $s_0 \in \cS$. Denoting by $\Pi$ some set of policies, of which we assume at least one is feasible, the goal is to solve
\begin{align*}
	\max_{\pi \in \Pi}  &~J(\pi; \Kb) := ~\mathbb{E}^{\PP}\Biggl[\sum_{t=0}^{\infty} \gamma^t R_t \Biggr] \numberthis \label{eqn:obj-mdp}      \\
	\text{subject to: }     &  \\
	& \EE^{\PP}\left[g^m_{t}(s_t) \right] \leq K_{t,m}, ~ \forall t \in \ZZ^{\geq 0}, m \in [1,M].
\end{align*}
Note that although we selected a specific information model for the constraints, \eqref{eqn:obj-mdp} is quite general and any classical constrained MDP formulation can be recast similarly.  Depending upon the specific state space, action space and policy class, there are two types of approaches one might use to solve constrained MDPs:
\begin{itemize}
	\item linear programming (dynamic programming is not applicable in the constrained setting) or
	\item Lagrangian relaxation. 
\end{itemize}
We are interested in the second class of methods as linear programming approaches are not as widely applicable. Taking the Lagrangian relaxation we get
\begin{equation}
	\label{eqn:obj-mdp-lgr-primal}
	\cL(\pi;\blambda) := ~\mathbb{E}^{\PP}\Biggl[\sum_{t=0}^\infty \gamma^t R_t   + \sum_{t=0}^\infty \sum_{i=1}^M \lambda_{m,t} (K_{m,t} - g^m_{t}(s_t)) \Biggr].
\end{equation}
The optimization problem then becomes
\begin{align*}
	\min_{\blambda} \max_{\pi \in \Pi}~  &\cL(\pi;\blambda)  \numberthis \label{eqn:obj-mdp-lgr}      \\
	\text{subject to: }     &  \\
	& \blambda \geq 0,
\end{align*}


\subsection{Standard Approach: Model Predictive Control (MPC)}
A standard way to solve \eqref{eqn:obj-mdp-lgr} is with {\it model predictive control}. The model predictive control procedure takes the current state $s_t$, a model of the system dynamics $\hat{\PP} \approx \PP$ and a constraint values $\Kb$ as input, and produces the next action $a_t$.
\paragraph{Adaptive MPC}
In adaptive model predictive control, one estimates a model of the system dynamics $\hat{\PP}$ (based on the available data up through time $t$) and then solves for the optimal sequence of actions over the next $H$ periods by assuming $\hat{\PP}$ is the true model. Formally, at time $t$ starting from state $s_t$, the MPC policy optimizes the next $H$ actions directly by solving 
\begin{align*}
	\max_{a_t,\dots,a_{t+H}}  &~\mathbb{E}^{\hat{\PP}}\Biggl[\sum_{s=t}^{t+H} \gamma^t R_s \Biggr] \numberthis \label{eqn:obj-mpc}      \\
	\text{subject to: }     &  \\
	& \EE^{\hat{\PP}}\left[g^m_{t}(s_t) \right] \leq K_{t,m}, ~ \forall t \in [0,T], m \in [1,M].
\end{align*}
where the objective function and state evolution correspond to the MDP $\hat{\cM} := (\cS,\cX,\hat{\PP},R,\gamma,H,s_t)$. An MPC solution method may solve the lagrangian relaxation instead:
\begin{align*}
	\min_{\hat{\blambda}} \max_{a_t,\dots,a_{t+H}}~  & \mathbb{E}^{\hat{\PP}}\Biggl[\sum_{s=t}^{t+H} \gamma^{s-t} R_s   + \sum_{s=t}^{t+H} \sum_{i=1}^M \hat{\lambda}_{m,s} (K_{m,s} - g^m_{s}(s_s)) \Biggr]  \numberthis \label{eqn:obj-mpc-lgr}      \\
	\text{subject to: }     &  \\
	& \blambda \geq 0.
\end{align*}
Denoting by $a^*_{t,0},\dots,a^*_{t, H}$ the primal solution to either \eqref{eqn:obj-mpc} or \eqref{eqn:obj-mpc-lgr}, the MPC policy is defined as $\pi^{MPC}(s_t) := a^*_{t,0}$. 

\paragraph{Policy-Based MPC}
Along the lines of \eqref{eqn:obj-mpc-lgr}, another common approach is to solve \eqref{eqn:obj-mpc-lgr} but with a restricted policy class. For example, in an inventory control setting (see \Cref{sec:idp}), practitioners may use a newsvendor (or base-stock) policy. In that case, one typically augments the state representation from \eqref{eqn:obj-mpc} with the dual costs over the planning horizon, $\blambda \in \RR^{H\times M}_{\geq 0}$, and the model predictive control problem becomes
\begin{align*}
	\min_{\hat{\blambda}} \max_{\pi \in \Pi} ~  & \mathbb{E}^{\hat{\PP}}\Biggl[\sum_{s=t}^{t+H} \gamma^{s-t} R_s   + \sum_{s=t}^{t+H} \sum_{i=1}^M \hat{\lambda}_{m,s} (K_{m,s} - g^m_{s}(s_s)) \Biggr]  \numberthis \label{eqn:obj-mpc-lgr-2}      \\
	\text{subject to: }     &  \\
	& \hat{\blambda} \geq 0,
\end{align*}
where the actions are sampled from some policy $\pi: \cS \times \RR^{H\times M}_{\geq 0} \rightarrow \Delta(\cX)$. Note that \eqref{eqn:obj-mpc} is a strict generalization of the case where we have a fixed policy $\pi$ and are performing dual descent on the constraint costs.

\begin{remark}[Including dual costs as part of the state representation] Technically, one does not need to include $\blambda$ as part of the state representation because the formulation above considers only a single $\Kb$. In practical settings, however, the constraint at period $t$ may only be known a $H$ periods beforehand. To use the IDP in \Cref{sec:idp} as an example, a retailer might build or purchase additional storage space if they expected to be overly constrained at time $t$. By explicitly including $\blambda$ as part of the state, the idea would be that the policy can handle different constraint settings $\Kb$.
\end{remark}

\subsection{Our Approach: Cost Forecasting}
Our proposed approach can be viewed as forecasting the dual costs $\blambda$ produced by an MPC that has access to the true MDP $\cM$ rather than an estimate $\hat{\cM}$ -- in the terminology of \Cref{sec:idp}, this would be forecasting the costs produced by a dual cost search against historic realizations. Our procedure requires access to $\cM$ and the ability to reset to any state $s$ (in order to perform the cost search).  To make this precise, we are interested in the solution to the following optimization
\begin{align*}
	\min_{\blambda^*} \max_{\pi \in \Pi} ~  & \mathbb{E}^{\PP}\Biggl[\sum_{s=t}^{t+H} \gamma^{s-t} R_s   + \sum_{s=t}^{t+H} \sum_{i=1}^M \lambda^*_{m,s} (K_{m,s} - g^m_{s}(s_s)) \Biggr]  \numberthis \label{eqn:obj-mpc-exact}      \\
	\text{subject to: }     &  \\
	& \blambda^* \geq 0.
\end{align*}
Many algorithms could be used to solve the inner optimization, including RL by incorporating the constraint penalties into the reward function. In the remainder of this section we will denote the optimization solved in \eqref{eqn:obj-mpc-exact} as $\cP(\Kb', s, t)$ where $\Kb'$ are the constraints (e.g. $\Kb_{t:t+H,:}$) in $\RR^{H\times M}$, $s \in \cS$ is the initial state and $t$ is the starting time. Let $\blambda^*_{s,:}$ denote the component of the optimal solution to \eqref{eqn:obj-mpc-exact} corresponding to time $s\in [t,t+H]$. The key idea is at each time $t$ to produce a probabilistic forecast  given the current state and the constraints over the next $H$ periods:
\begin{equation}
	\label{eqn:cost-forecast}
	p(\blambda^*_{t,:},\dots,\blambda^*_{t+H,:} | s_t, \Kb_{t:t+H,:}).
\end{equation}
Observe that \eqref{eqn:cost-forecast} is combining two forecasting problems in a single end-to-end forecast. By forecasting $\blambda^*$ directly, there is no need to forecast the dynamics $\PP$ nor solve the outer optimization in \eqref{eqn:obj-mpc-exact}.

\subsubsection*{Learning Procedure}
To present the learning procedure, we use a simplified form of \eqref{eqn:cost-forecast} where the goal is to predict a single summary statistic (e.g. the mean) of the distribution in \eqref{eqn:cost-forecast} over the next $H$ periods. Additionally, let $\pi^{fixed}: \cS \times \RR^{H\times M}_{\geq 0} \rightarrow \Delta(\cX)$ be a fixed policy and $\Pi := \{ \pi^{fixed} \}$ the policy class in \eqref{eqn:obj-mpc-exact}.

Now, consider a class of regression models parameterized by  $\omega \in \Omega$, $\phi_{\omega} : \cS \times \RR^{H\times M} \times \rightarrow \RR^H$. Let $\cA^{sup}$ denote a supervised learning procedure that given a dataset $\cD := \{ (s^i,\Kb^i,\blambda^i) \}$ produces a model $\omega$ that minimizes a loss $l : \cS  \times \RR^{H\times M} \times \RR_{\geq 0}^{H\times M} \rightarrow \RR $ on the dataset $\cD$. Additionally, let $P^G$ be a distribution over the space of possible constraint sequences. The simplest thing to do would be to sample some constraint sequences from a distribution $P^K$, run rollouts under $\pi^{fixed}$ to generate a dataset $\cD_0$ and then obtain $\omega_0 = \cA^{sup}(\cD_0)$. However when the predictive model $\omega_0$ is used to produce costs used as inputs to $\pi^{fixed}$ in $\cM$, it may not properly adhere to constraints. The issue is that $\omega_0$ is fit with supervised learning, but the data generating process in $\cM$ under $\pi^{fixed}$ and $\omega_0$ is not i.i.d. The same problem also occurs in imitation learning \citep{ross2011reduction}, and \Cref{alg:ncc-generic} is inspired by the DAGGER algorithm from the imitation learning literature. 

\begin{algorithm}[H]
	\caption{Neural Coordinator}
	\label{alg:ncc-generic}
	\begin{algorithmic}
		\State \textbf{Inputs: } $\omega'$,  $P^K$, $T$, $\pi: \cS \times \RR^{H\times M}_{\geq 0} \rightarrow \Delta(\cX)$, $\cA^{sup}$, horizon $H$
		\State $\cD^{0} \gets \emptyset$
		\State $n \gets 0$
		\State $\omega^0 \gets \omega'$
		\State \texttt{\# Keep updating until converged}
		\While{not converged}
			\State $n \gets n+1$
			\State Sample $\Kb^n \sim P^K$
			\State Run $T$-step rollout in $\cM$ under $\pi$ and $\omega^{n-1}$
			\State From rollout, construct $\cU^n := \{(0, s^n_0, \Kb^n_{t:t+2H:, }),\dots,(T, s^n_T, \Kb^n_{T:T+2H:, })\}$
			\State \texttt{\# Augment existing dataset with states from current trajectory}
			\For{ $(t, s^j, \Kb^j) \in \cU^n$}
				\State $\lambda^{j,*} \gets \cP(\Kb^j, s^j, t)$ with policy class $\Pi =\{ \pi \}$.
				\State $\cD^n \gets \cD^{n-1} \cup \{ (s^j, \Kb^j, \blambda^{j,*}) \} $
			\EndFor
			\State  \texttt{\# Update coordinator}
			\State $\omega^n \gets \cA^{sup}(\cD)$
		\EndWhile
		\State \textbf{Returns: } $\omega^n$
	\end{algorithmic}
\end{algorithm}

\paragraph{Non-stationary MDPs} Even if we do not have access to the actual MDP, the approach described above can still be desirable. If the dynamics are non-stationary, as is often encountered in the real-world, we may have a much better estimate of the dynamics for times $s \leq t$ than we do for $s > t$, where $t$ is the current time. For example, in the exo-MDP framework \citep{sinclair2023hindsight}, one might have access to samples from the actual noise process for $s \leq t$.

\section{IDP Construction, Capacity Constraints and Coordination Mechanisms}
\label{sec:idp}

In this section, we follow the Interactive Decision Process (IDP) formulation of \citet{madeka2022deep}, borrowing most of the conventions and notation. A planner manages a set $\cA$ of products, and for each $i \in \cA$ and at each time step $t = 1,2,\dots,T$, the planner is trying to determine how many units of inventory to order to satisfy demands $D^i_t$. In this work, we consider the addition of {\it capacity constraints} into the problem formulation. At each time $t$, there is a constraint $K^1_T$ on storage capacity and another constraint $K^2_T$ on arrivals (or inbound) capacity.

This section is organized as follows: In \Cref{sec:idpconstruction}, starting with the formulation in \citet{madeka2022deep}, we add capacity constraints to the problem formulation. Then in \Cref{sec:idpconstruction-lagrangian} we  formally define the notion of a {\it coordination mechanism} that will be used in solving a Lagrangian relaxation of the constrained problem. Finally, we define the Lagrangian IDP which considers a problem similar to the inner optimization of the Lagrangian relaxation.

\subsection{Capacity Constrained IDP}
\label{sec:idpconstruction}

Our IDP is governed by external (exogenous) processes, a control process, inventory evolution dynamics, a sequence of constraints, and a reward function. The inventory management problem seeks to find the optimal inventory level for each product $i$ in the set of retailer's products, which we denote by $\mathcal{A}$. We assume our exogenous random variables are defined on a canonical probability space $(\Omega, \mathcal{F}, \mathbb{P})$, and policies are parameterized by $\theta$ in some parameter set $\Theta$. Additionally, there are constraints on total storage and total inbound at each time step $t$, denoted $K^1_t$ and $K^2_t$, respectively. The constraints are known at time $t=0$ and are exogenous to the buying policy. 

\subsubsection*{External Processes}
At every time step $t$, for product $i$, there is a random demand process $D_t^{i} \in [0,\infty)$ that corresponds to customer demand during time $t$ for product $i$. The random variables $p_t^{i} \in [0,\infty)$ and $c_t^{i} \in [0,\infty)$  correspond to selling price and purchase cost. Finally, the random variable $v_t^i \in \ZZ_{\geq 0}$ denote vendor lead time for an order placed at period $t$ for item $i$. Each product also has a {\it storage weight} and {\it inbound weight} that are used to weight the contribution of units of product $i$ to consumption of the shared constraints (for example these might be volumes, or time needed to process an arrival of that product). We denote these weights by $w^i$ and $u^i$, respectively.  Our exogenous state vector for product $i$ at time $t$, which takes values in $\cS$, is all of this information\footnote{For ease of exposition, the non-time-varying $w^i$ and $u^i$ are included in the state vector at each time step.}:
\[
	s^i_t := (D_t^{i}, p_t^{i}, c_t^{i}, v_t^i, w^i, u^i) \in \cS.
\]
To allow for the most general formulation possible, we consider policies that can leverage the history of all products. Therefore, we will define the history 
\[
	H_t :=\{( s_1^i,\dots,s_{t-1}^i)\}_{i=1}^{|\cA|}
\]
as the joint history vector of the external processes for all the products.

\subsubsection*{Control Processes}
Our control process will involve, at each time $t$, picking actions $\ab_t \in \mathbb{R}_{\geq 0}^{|\mathcal{A}|} $ for each product jointly. For product $i$, the action taken is the order quantity for product $i$ and denoted by $a_t^{i} \in \mathbb{R}_{\geq 0}$.

\subsubsection*{Inventory Evolution}
We assume that the implicit endogenous inventory state follows standard inventory dynamics and conventions. Inventory arrives at the beginning of the time period, so the inventory state transition function is equal to the order arrivals at the beginning of the week minus the demand fulfilled over the course of the week. The number of units arriving during the period is denoted
\begin{equation}
\label{eqn:invarrivals}
J_t^i = \sum_{0 \leq k < t} a^i_k \bbone_{v^i_k = (t-k)}
\end{equation}
and the inventory update is therefore
\begin{equation}
\label{eqn:invupdate}
I^i_{t_-} = I^i_{t-1} + J_t^i,
\end{equation}
where $I_t^i$ is the inventory at the end of time $t$, and $I_{t_-}^i$ is the inventory at the beginning of time $t$, after arrivals but before demand is fulfilled. Then, at the end of time $t$, the inventory position is:
\[
I^i_{t} = \min(I^i_{t_-}- D^i_t, 0).
\]
We additionally define the {\it aggregate} inventory and inbound as 
\[
	\tilde{I}_t := \sum_{i \in \cA} w^i I_t^i	~~\text{ and }~~ \tilde{J}_t := \sum_{i \in \cA} w^i J_t^i.	
\]

\subsubsection*{Constraints}
The constraint sequence, denoted as
\[
  G := \{K^1_1,K^2_1,\dots,K^1_{T},K^2_{T}\},
\]
is known to the decision maker at time $t=0$. Following the approach taken in the literature, we consider constraints in expectation at each time $t$
\[
K_t^1  \geq \EE^{\PP}\left[\tilde{I}_t\right] ~~\text{ and }~~ K_t^2 \geq\EE^{\PP} \left[\tilde{J}_t\right]. 
\]

\subsubsection*{Reward Function}
The reward at time $t$ for product $i$ is defined as the selling price times the total fulfilled demand, less the total cost associated with any newly ordered inventory (that  will be charged by the vendor upon delivery):
\begin{equation}
\label{eqn:rewardfunc}
    R_t^{i} = p_t^i \min(D^i_t, I^i_{t_-}) -  c_t^i a_t^i.   
\end{equation}

\subsubsection*{Policy Class}
For a class of policies parameterized by $\theta$, we can define the actions as
\[
	a_t^{i} = \pi_{\theta,t}^{i}(H_t, G).
\]
We characterize the set of these policies as $\Pi = \{\pi_{\theta,t}^{i}| \theta \in \Theta, i \in \mathcal{A}, t \in [0,T]\}.$

\subsubsection*{Optimization Problem}
We will write $R^i_t(\theta)$ to emphasize that the reward is a function of the policy parameters $\theta$. Recall that selling price and buying cost are determined exogenously. We assume all rewards $R_t^{i} \in [R^{min},R^{max}]$, and assume a multiplicative discount factor of $\gamma \in [0,1]$ representing the opportunity cost of reward to the business. Again, we make the dependence on the policy explicit by writing $R_t^{i}(\theta)$. The objective is to select the best policy (i.e., best $\theta \in \Theta$) to maximize the total discounted reward across all products, expressed as the following optimization problem:
\begin{align*}
	\max_{\theta}  &~\mathbb{E}^{\PP}\Biggl[\sum_{i\in \mathcal{A}} \sum_{t=0}^T \gamma^t R_t^{i}(\theta)\Biggr] \numberthis \label{eqn:obj-cc}      \\
	\text{subject to: }     &  \\
K_t^1  &\geq \EE^{\PP}\left[\tilde{I}_t\right] \\
K_t^2 &\geq\EE^{\PP}\left[\tilde{J}_t\right].
\end{align*}
with state transition dynamics governed by
\begin{align*}
	I^i_0                   & = k^i  \\
	a^i_t                   & = \pi_{\theta,t}^i(H_t, G) \\
	J_t^i &= \sum_{0 \leq k < t} a^i_k \bbone_{v^i_k = (t-k)} \\
	I^i_{t_-} &= I^i_{t-1} + J_t^i \\
	I^i_{t} & = \min(I^i_{t_-}- D^i_t, 0).
\end{align*}
Here, $\mathbb{P}$ denotes the joint distribution over the exogenous processes. The inventory $I_0^i$ is initialized at $k_i$, a known quantity \emph{a priori}.

\subsection{Approximating the Constrained IDP: Lagrangian IDP}
\label{sec:idpconstruction-lagrangian}

While problem \eqref{eqn:obj-cc} is what we want to solve, in general it is not possible to solve it directly. Instead we consider a Lagrangian relaxation \citep{boyd2004convex} 
\begin{equation}
	\label{eqn:lagrange-obj}
	\min_{\blambda^1, \blambda^2 \geq 0} \max_{\theta}  ~\mathbb{E}^{\PP}\Biggl[\sum_{i\in \mathcal{A}} \sum_{t=0}^T \gamma^t R_t^{i}(\theta) + \sum_{t=0}^T\lambda^1_t(K_t^1-\tilde{I}_t) + \sum_{t=0}^T\lambda^2_t(K_t^2 - \tilde{J}_t) \Biggr],
\end{equation}
which restores the separability across products.

In the remainder of this section, we describe a multi-agent interactive decision process wherein one agent sets prices on the constrained resources ($\blambda^1$, $\blambda^2$). Specifically, we construct the penalized IDP that corresponds to the Lagrangian in \eqref{eqn:lagrange-obj}. This can be thought of as an $|\cA|+1$ agent problem, where there are $|\cA|$ product level policies, and one coordinator agent that sets capacity prices. We describe only the components that differ from the constrained problem.

\paragraph{Coordination Mechanisms}
At each time $t$, a coordination mechanism sets two prices $\blambda_t := (\lambda_{t,0}^1,\lambda_{t,0}^2)  \geq 0$, corresponding to $K_t^1$ and $K_t^2$. The coordinator also announces future prices over the next $L$ periods, which we denote by $\hat{\blambda}^L_t$ as they can be viewed as a forecast the price. Formally, an $L$-step {\it coordination mechanism} is a sequence of vector-valued functions $F_t$ that maps a constraint sequence and joint product history to the next $L$ shadow prices
\[
	F_t(G, H_t) = (\lambda_{t,0}^1, \lambda_{t,0}^2,\hat{\lambda}_{t,1}^1, \hat{\lambda}_{t,1}^2,\dots,\hat{\lambda}_{t,L}^1, \hat{\lambda}_{t,L}^2).
\]
The coordination mechanism could be as simple as deterministically announcing the dual variable values in \eqref{eqn:lagrange-obj} if one uses an ADMM-style algorithm to solve the constrained optimization (see \Cref{rem:pen-reward} below). Other candidate mechanisms are model predictive control and neural cost forecasting proposed in \Cref{sec:ncc-generic}.

\subsubsection*{Control Processes} The history of prices announced up until time $t$ is defined as $H_t^{\lambda} := \{\lambda_{s,l}^j ~~|~~   j \in \{1,2\},~l \in \{0,\dots,L\},~s \leq t \}$. At the product level, for a class of policies parameterized by $\theta$, we can define the actions as
\[
	a_t^{i} = \pi_{\theta,t}^{i}(H_t,H^\lambda_t).
\]
We characterize the set of these policies as $\Pi = \{\pi_{\theta,t}^{i}| \theta \in \Theta, i \in \mathcal{A}, t \in [0,T]\}.$

\subsubsection*{Reward Function} We modify the previous reward \eqref{eqn:rewardfunc} to incorporate a penalty according to the prices set by the coordination mechanism 
\begin{equation}
\label{eqn:rewardfunc-pen}
    R_t^{\lambda, i} = R_t^i - \lambda_t^1 w^i I_{t}^i - \lambda_t^2 u^i J_{t}^i . 
\end{equation}

\subsubsection*{Optimization Problem}
As above we write $R^{\lambda,i}_t(\theta)$ to emphasize that the reward is a function of the policy parameters $\theta$. Given a coordination mechanism, we solve the following optimization problem:
\begin{equation}
	\label{eqn:lagrangian-idp}  
	\max_{\theta}  ~ J^{pen}(\theta, G) := \mathbb{E}^{\PP}\Biggl[\sum_{i\in \mathcal{A}} \sum_{t=0}^T \gamma^t R_t^{\lambda, i}(\theta)  \Biggr],	
\end{equation}
with dynamics governed by
\begin{align*}
	I^i_0                   & = k^i  \\
	(\lambda_t^1, \lambda_t^2,\hat{\lambda}_{t,1}^1, \hat{\lambda}_{t,1}^2,\dots,\hat{\lambda}_{t,L}^1, \hat{\lambda}_{t,L}^2) &= F_t(G, H_t) \\
	a^i_t                   & = \pi_{\theta,t}^i(H_t, H^\lambda_t) \\
	J_t^i &= \sum_{0 \leq k < t} a^i_k \bbone_{v^i_k = (t-k)} \\
	I^i_{t_-} &= I^i_{t-1} + J_t^i \\
	I^i_{t} & = \min(I^i_{t_-}- D^i_t, 0).
\end{align*}
Having defined the optimization problem, it is now clear why we constructed the penalized reward \eqref{eqn:rewardfunc-pen}.
\begin{remark}[Motivating the penalized reward]
	\label{rem:pen-reward}
Let 
\[\blambda' := (\lambda^{1,'}_{0},\lambda^{2,'}_{0},\dots\lambda^{1,'}_{T},\lambda^{2,'}_{T})\] be a some sequence of costs. We can define a ``teacher-forcing'' coordination mechanism for optimizing \eqref{eqn:lagrangian-idp} as
\begin{equation}
	\label{eqn:teacher-force}
	F_t(G, H_t) := (\lambda^{1,'}_{t}, \lambda^{2,'}_{t},\lambda^{1,'}_{t}, \lambda^{2,'}_{t}\dots,\lambda^{1,'}_{t+L},\lambda^{2,'}_{t+L}).
\end{equation}
Under this specific coordination mechanism, in an abuse of notation we can write $J^{pen}(\theta, \blambda')$. We see that $\max_{\theta} J^{pen}(\theta, \blambda')$ is {\it equivalent up to constants} to the inner optimization in \eqref{eqn:lagrange-obj}.
\end{remark}
Given a fixed $G$, an ADMM-like \citep{boyd2011foundations} procedure can be used to alternate updating the solution to the inner optimization problem and the dual prices $\blambda'$ and one can augment the Lagrangian objective for stability (see \Cref{sec:alg-details}).

\subsection{Neural Coordinator}
\label{sec:neural-coord}
One example of a coordination mechanism is model predictive control, which would use forecasted demand to perform a dual cost search for the next $L$ periods. However, for an RL buying policy that uses many historical features, model predictive control would require forward simulating many features and it is difficult to accurately forecast the full joint distribution of all these features. 

Instead, we follow the approach described in \Cref{sec:ncc-generic} and propose to use a neural network to {\it forecast} the capacity costs that would be required to constrain the policy. This forecaster then can be used as an $L$-step coordination mechanism.  Specifically, we learn a neural network parameterized by $\omega \in \Omega$, $\phi_{\omega,t}(H_t, G)$, that outputs 
\[
	\phi_{\omega,t}(H_t, G) = (\lambda_{t,0}^1, \lambda_{t,0}^2, \dots,\lambda_{t,L}^1,\lambda_{t,L}^2)
\]
where $\lambda_{t,l}^j$ is a forecast of the final cost $\lambda_{t+l,0}^j$.

Our proposed approach is to use \Cref{alg:ncc-generic} to predict the behavior of a dual cost search. In \Cref{sec:neural-coordinator}, we describe an approximation to \Cref{alg:ncc-generic} that we use in our empirical work and which is able to leverage the DirectBackprop algorithm to optimize the coordinator's decision.

\section{Capacity Curve Sampling}
\label{sec:constraint-sampling}

Because we only observe one sample historically of Amazon's capacity in any given marketplace, if we backtest against a capacity control mechanism against this single capacity curve, we would not be able to obtain any sort of generalization guarantee. For example, in a world where Amazon did not offer one day shipping, the network capacity would have looked different through time.  In order for a capacity control mechanism to be useful in practice, we need to ensure that in many scenarios, the coordination mechanism will properly constrain. Intuitively, what we would like to backtest against is a large set of capacity curves that (in different worlds) Amazon might have had instead. 

One salient property of real-world capacity curves is that they tend to have discontinuities as capacity comes online or goes offline. In order to sample paths from some space, we must first decide on a choice of function space and a basis on that space. \citet{donoho1992unconditional} showed that wavelet bases are optimal for representing functions that have arbitrary discontinuities. \Cref{fig:example-paths} shows examples of generated constraint paths.

\begin{figure}[h]
	\centering
	\amazon{\includegraphics[width=0.4\textwidth]{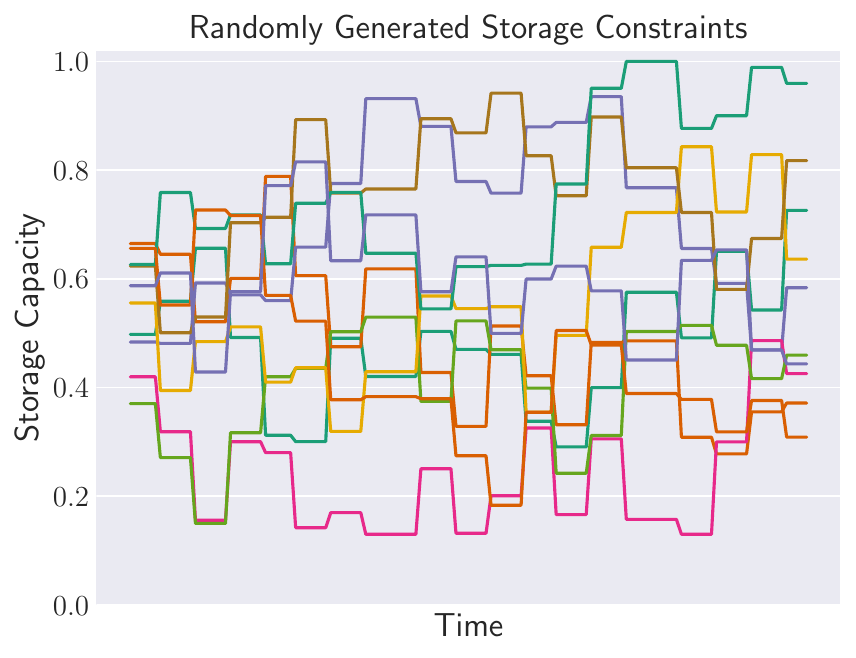} ~~ \includegraphics[width=0.4\textwidth]{plots/css_acc_sortld_final.pdf}}{\includegraphics[width=0.4\textwidth]{plots/path_examples_normalized.pdf} ~~ \includegraphics[width=0.4\textwidth]{plots/ext_sortld_storage_stylized.pdf}}
	\caption{\amazon{Example of constraint paths generated by our proposed sampling scheme on the left, and an actual constraint path on the right.}{Example of constraint paths generated by our proposed sampling scheme on the left, versus a stylized example on the right.}}
	\label{fig:example-paths}
\end{figure}

\subsubsection*{Haar Wavelet Basis}
The Haar wavelet function is defined as 
\begin{equation}
	\label{eqn:haar-wavelet}
	\psi(x) = \begin{cases}
		1 & \text{ if } 0 \leq x < \frac{1}{2} \\
		-1 & \text{ if } \frac{1}{2} \leq x < 1 \\
		0 & \text{ o.w.}
	\end{cases}
\end{equation}
Then, for every $n,k \in \ZZ$ the Haar function is defined as
\begin{equation}
\label{eqn:haar-function}
\psi_{n,k}(t) = 2^{n/2}\psi(2^nt - k).
\end{equation}
We are interested only in functions supported on $[0,1]$, so we consider the Haar System on $[0,1]$ \citep{haar1910zur} which is the set of wavelets
\begin{equation}
\{ \psi_{n,k} : n \in \ZZ_{\geq 0}, 0 \leq k < 2^n \},
\end{equation}
and it is an orthonormal basis of $L^2([0,1])$. 

\subsubsection*{Constraint Space}
Mathematically, a {\it capacity curve} is function $f: [0,1] \rightarrow \RR_{\geq 0}$. First, we define a modified Haar wavelet
\begin{equation}
	\label{eqn:haar-function-mod}
	\tilde{\psi}_{n,k}(t) = \psi(2^nt - k),
\end{equation}
which gives a basis that is orthogonal (but not normalized). This will allow us to sample in a manner where it is simple to control the total variation of the generated paths. Then we define the Haar basis of order $m$ as
\begin{equation}
	\label{eqn:haar-basis-m}
	\cH_m = \Bigl\{ \tilde{\psi}_{n,k} : n \in \{0,\dots,m\}, 0 \leq k < 2^n \Bigr\}.
\end{equation}

\subsubsection*{Sampling Scheme}
Our sampler is characterized by the following: (1) the order $m$ of the Haar basis, (2) a scale $\nu$ that controls the total variation of the generated paths. To obtain paths on $[1,T]$, we sample from a space of paths supported on [0,1] and then rescale to obtain a path on $[1,T]$. Specifically, in \Cref{sec:experiment} we sample coefficients from $\cN(0,\nu/(2^{m+1} - 1))$.

\section{Learnability Results}
\label{sec:learn}

Now we turn to the question: can we backtest policy performance even in the presence
of capacity constraints and a capacity control mechanism? As in \citet{madeka2022deep},
we assume that the products are independent -- i.e. that $\PP$ is a product measure
$\PP = \prod_i \PP^i$ and that actions are determined fully by exogenous random variable,
the policy parameterized by $\theta \in \Theta$ and the coordinator,
parameterized by $\omega \in \Omega$.

In the uncapacitated setting, the total reward is a sum of independent random variables,
so a Hoeffding bound provides a strong generalization guarantee as long as the number of
products is large. The existence of capacity constraints couple all the products, breaking this approach.
However, it turns out that we can still provide a generalization bound, as long as the
coupling between products induced by the constraints and coordinator are \emph{weak}.

To make this formal in \Cref{asm:mcd_asm}, we define 
\[
	\bar{H}^i_t :=  (H_t \setminus \{(\bar{s}_1^i \dots \bar{s}_{t-1}^i)\}) \cup \{(s_1^i \dots s_{t-1}^i)\}
\]
as the history with the $i$th exogenous series modified using $(\bar{s}_1^i \dots \bar{s}_{t-1}^i) \in \bigotimes_{s=1}^{t-1} \cS$. Further, for any constraint path $G$, any policy $\theta \in \Theta$ and any
coordinator $\omega \in \Omega$, we will denote the combined policy (buying agent and coordinator) at time $t$ as $\ba_t = \pi_{\theta, \omega}(H_t, G)$. 

\begin{assumption}[Single-product Robustness]
	\label{asm:mcd_asm}
For any product $i$ and any modification $(\bar{s}_1^i \dots \bar{s}_{t-1}^i) \in \bigotimes_{s=1}^{t-1} \cS$ there exists a constant $c_a$ such that for all $t$, $\|\pi_{\theta, \omega}(H_t, G) - \pi_{\theta, \omega}(\bar{H}^j_t, G)\|_1 \leq c_a$.
\end{assumption}

Informally, this assumption assures us that the coordinator is reasonably well-behaved.
If we counterfactually change a single product, we expect at most an $O(1)$
change in capacity usage. Suppose extra capacity is occupied by this one product, then
$O(1)$ change in aggregate ordering across all products is required to meet the capacity.

We use reward $R_t^{i}(\theta, \omega, H_t, G)$ as a notational convenience for the reward as a function of $\theta$ and $\omega$, the exogenous history $H_t$ and
the capacity curves $G$. Given a curve $G$, we are interested in measuring
\begin{equation*}
V^i_T(\theta, \omega, G) := \EE^{\PP} \left[\sum_{t=1}^T \gamma^t R^i_t(\theta, \omega, H_t, G)\right],
\end{equation*}
with the corresponding population-level objective given by
\begin{equation*}
	V_T(\theta, \omega, G) := \frac{1}{|\cA|} \sum_{i\in\cA}V^i_T(\theta, \omega, G).
\end{equation*}
The notation $\hat{V}^{i}_T(\theta, \omega, G)$ and $\hat{V}_T(\theta, \omega, G)$ denote the empirically 
sampled versions.

We also are interested in the expected constraint violation
and sampled constraint violation defined as
\begin{equation*}
C^1_T(\theta, \omega, G) := \EE^{\PP} \left[\sum_{t=1}^T \left(\sum_{i \in \mathcal{A}} w^i I_t^i - K_t^1 \right)_{+}\right],
\end{equation*} and
\begin{equation*}
C^2_T(\theta, \omega, G) := \EE^{\PP} \left[\sum_{t=1}^T \left(\sum_{i \in \mathcal{A}} u^i J_t^i - K_t^2 \right)_{+}\right],
\end{equation*}
with sampled constraint violations $\hat{C}^1_T(\theta, \omega, G)$ and $\hat{C}^2_T(\theta, \omega, G)$ defined
analogously. \Cref{thm:learn-reward-v2} shows that we can efficiently
backtest any policy $(\theta, \omega) \in \Theta \times \Omega$ on any constraint path $G \in \mathcal{G}$, where $\mathcal{G}$
is a set of constraint paths of interest\footnote{We can approximately cover all paths with $O(\exp(T))$ paths.}.
\begin{theorem}
	\label{thm:learn-reward-v2}
Let $\Theta$ and $\Omega$ be finite sets and let \Cref{asm:mcd_asm} hold. Given any $\delta \in (0,1)$, with probability greater than $1-\delta$
	we have that for all $(\theta, \omega) \in \Theta \times \Omega$:
\begin{align*}
	\max_{G \in \mathcal{G}}|\hat{V}_T(\theta, \omega, G) - V_T(\theta, \omega, G)| \leq & c_a(p_{\max} + c_{\max})T^2\left( \sqrt{\frac{\log\left(2|\Theta||\Omega||\mathcal{G}|/\delta\right)}{2|\cA|}}\right)\\
	\max_{G \in \mathcal{G}}|\hat{C}^1_T(\theta, \omega, G) - C^1_T(\theta, \omega, G)| \leq & c_a w_{\max}T^2\left( \sqrt{\frac{|\cA|\log\left(2|\Theta||\Omega||\mathcal{G}|/\delta\right)}{2}}\right)\\
	\max_{G \in \mathcal{G}} |\hat{C}^2_T(\theta, \omega, G) - C^2_T(\theta, \omega, G)| \leq & c_a u_{\max}T^2\left( \sqrt{\frac{|\cA|\log\left(2|\Theta||\Omega||\mathcal{G}|/\delta\right)}{2}}\right)
\end{align*}
	where $p_{\max},c_{\max}, w_{\max}, u_{\max}$ are maximal prices, costs, and weights respectively.
\end{theorem}
\begin{proof}
	We will prove this using McDiarmid's inequality. In order to do so, we need to translate
	\Cref{asm:mcd_asm} into an upper bound on the change in reward induced by a change of action
	of size $c_a$. We would like to compare trajectories resulting from $H_t$ and $\bar{H}^i_t$. We note
	that a stock-out only reduces the deviation between two inventory sequences. As such, we can upper bound
	the deviation in reward purely in terms of the total inventory purchased:

	\[|R^j_t(\theta, \omega, H_t, G) - R^j_t(\theta, \omega, \bar{H}_t, G)| \le (p^j_{\max} + c^j_{\max}) \cdot \sum_{t=1}^T |a^j_t - \bar{a}_t^j|.\]
	Summing over all the products and times, and applying \Cref{asm:mcd_asm}, we get

	\[
		\Biggl|\sum_{t=1}^T \frac{1}{|\cA|} \sum_{j \in \mathcal{A}} R^j_t(H_t) - \sum_{t=1}^T \frac{1}{|\cA|}\sum_{j \in \mathcal{A}} R^j_t(\bar{H}^i_t)\Biggr| \le \frac{T^2 (p_{\max} + c_{\max})c_a}{|\cA|}.
	\]

	Now applying McDiarmid's inequality, we have
	\[
		\PP[|\hat{V}_T(\theta, \omega) - V_T(\theta, \omega)| > \epsilon] \leq 2 \exp \left(- \frac{2 \epsilon^2 |\mathcal{A}|}{ (T^2 (p_{\max} + c_{\max})c_a)^2} \right)
		 \]
	Setting this equal to $\frac{\delta}{||\Theta||\Omega||\mathcal{G}|}$ and combining with a union bound provides
	the first result. The average constraint violation bounds follow similarly by bounding the change in
	constraint violation in terms of the change in action.
\end{proof}

\section{Experimental Results}
\label{sec:experiment}
In this section we backtest our proposed methodology on a simplified form of the problem we impose no inbound constraint, and $w^i = 1$ for all $i \in \cA$. We use storage, rather than inbound constraints in these experiments because storage constraints are more challenging to plan against -- for example, it is easy to satisfy an inbound constraint by ordering zero, but storage must be carefully managed through time as inventory can only decrease through demand realization. 

\subsubsection*{Data}
For training both the buying and coordinating policies, we use a dataset of 250K products from the US marketplace from  June 2017 to February 2020. Our out-of-sample-backtest period is from May 2022 to May 2023, again on a population of 250K products.

\subsubsection*{Path Space}
The constraint paths are from a space of functions of bounded variation. These are represented using the Haar wavelet basis, and we sample the coefficients in that basis from a multivariate Gaussian. \Cref{fig:example-paths} shows examples of the paths (they are scaled up proportional to demand).

\subsubsection*{Policies}
\amazon{
	\input{internal_details/policies.tex}
}{
	We compare a reinforcement learning agent trained using direct backprop with a production base stock policy. Our proposed neural coordinator is used to constrain the RL agent, and a model predictive controller is used to constrain the base stock policy. 
}

\subsubsection*{Capacity Violation Metrics}
In addition to measuring reward achieved by the various policies, we consider several additional measures of constraint violation. They are (M1) mean constraint violation, (M2) mean violation on weeks where either unconstrained policy met at least 90\% of the limit, (M3) percent of weeks where the violation exceeded 10\%, and (M4) percent of weeks where constraint violation exceeded 10\% and either unconstrained policy exceeded 90\% of the capacity limit.

\subsubsection*{Results}
\Cref{tab:evaluation-oot} shows the results on the out of sample backtest, where each combination of policy and coordinator were evaluated against 100 storage constraint paths from $\PP^G$. Note that under all metrics (both violation and reward) the RL policy with neural coordinator outperforms the base stock baseline with model predictive control. Although some of the violation metrics seem somewhat high, one should keep in mind that many of the capacity curves sampled will be {\it highly constraining}, much more so than in a real-world setting -- in practice if the supply chain were so constrained, one would build more capacity. 

\begin{table*}[htp]
	\caption{Out-Of-Time evaluation for multiple initializations; all rewards are rescaled versus unconstrained base stock. For violations, lower is better; for reward, higher is better. }
	\label{tab:evaluation-oot}
	\begin{center}
		\small
		\begin{sc}
			\begin{tabular}{lllccccc}
				\toprule
				& 	  &         & \multicolumn{4}{c}{Violations} & \\
				Initialization &  Policy & Coordinator  & M1      & M2     & M3       & M4     & Reward  \\
				\midrule
                \multirow{4}*{\parbox{2.2cm}{Onhand \\w/ Inflight}} 
             &  RL   & --     & 31.4\%     & --   &  42.2\%    & --  &  {\bf 102.1}    \\
             &  RL & Neural   & {\bf 2.4\%}     &  {\bf 4.6\%}    &  {\bf 10.7\%}    & {\bf 20.1\%}    &  100.7    \\
				     &  base stock  & --     & 26.5\%     &  --    &  38.7\%    &  --   &  100.0         \\
             &  base stock  & MPC    & 5.3\%     &  10.1\%    &  17.6\%    &  33.1\%   &  99.1        \\
				\midrule
        \multirow{4}*{\parbox{2.2cm}{Zero}}
              &  RL   & --     & 24.8\%    & --      &  34.6\%    & --      &  {\bf 104.8}    \\
              &  RL & Neural   & {\bf 1.9\%}     & {\bf 4.3\%}   &  {\bf 8.5\%}     & {\bf 18.3\%}  &  103.1   \\
              &  base stock  & --     & 24.1\%    & --      &  35.1\%    & --      &  100.0         \\
              &  base stock  & MPC    & 4.8\%     & 10.1\%  &  16.2\%    & 34.2\%  &  99.3        \\
				\bottomrule
			\end{tabular}
		\end{sc}
	\end{center}
  \vskip -0.5em
\end{table*}

\Cref{fig:constraint-inv-paths-oot} below shows two examples of trajectories in the evaluation period for both policies and coordinators. We can see that the neural coordinator is succesfully able to constrain {\it out-of-sample}. See \Cref{sec:experiments-more} for additional results including a comparison of cost trajectories under the neural coordinator vs MPC -- in general the costs produced by the neural coordinator appear martingale, while those produced by MPC do not\amazon{ (similar to historic capacity costs -- see \Cref{sec:costs-acc-evo}).}{.}

\begin{figure}[h]
	\includegraphics[width=0.9\textwidth]{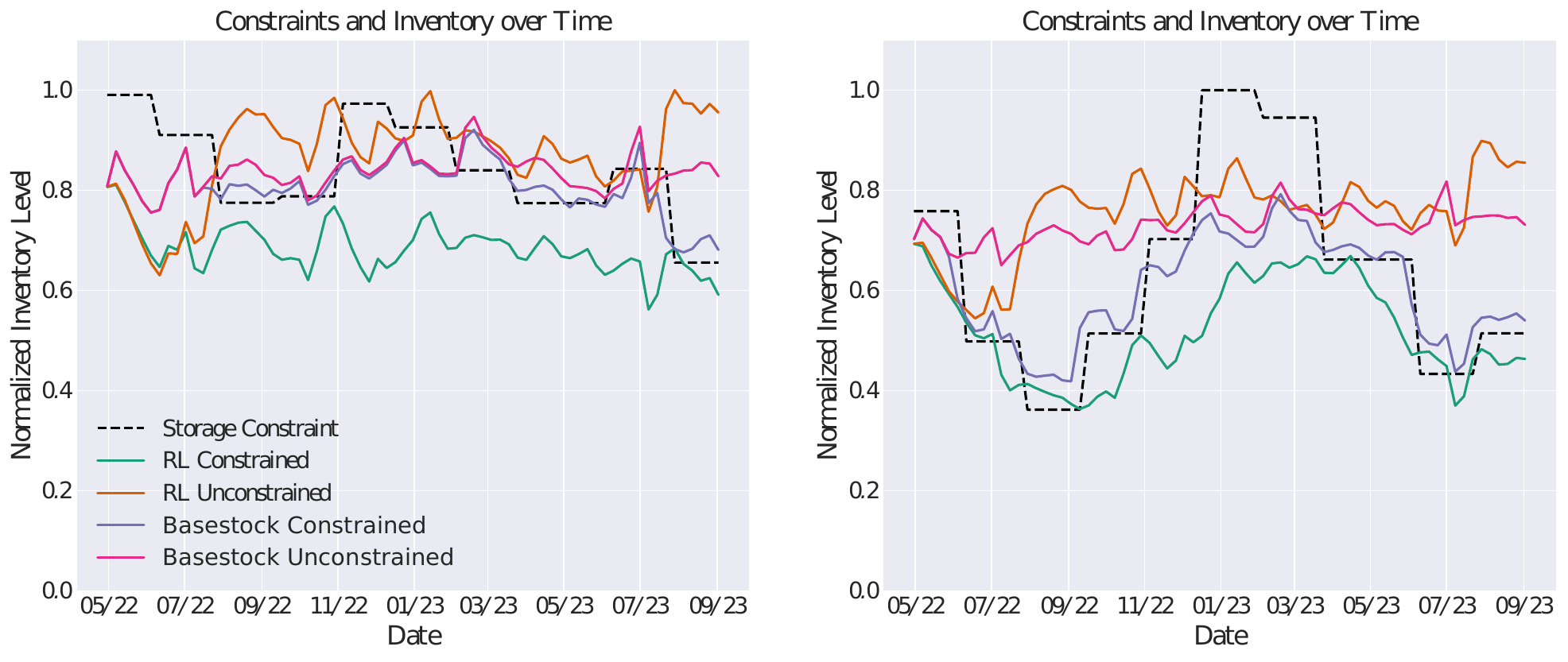}
	\caption{Inventory trajectories under different constraint paths during the evaluation period.}
  \vskip -1em
	\label{fig:constraint-inv-paths-oot}
\end{figure}

\section{Conclusion and Discussion}
We introduced a new approach to backtesting capacity control mechanisms and showed that our proposed neural coordinator can outperform a model predictive control baseline. Interesting directions of future work include alternative algorithms for learning the coordinator, how to handle evolution of capacity constraints \amazon{(see \Cref{sec:constraint-evo})}{(e.g. labor plans can change over time)}, and an evaluation against actual historic capacity curves.

\clearpage
\bibliographystyle{styles/ims_nourl_eprint}
\bibliography{external,internal}

\begin{thebibliography}{44}
\expandafter\ifx\csname natexlab\endcsname\relax\def\natexlab#1{#1}\fi
\expandafter\ifx\csname url\endcsname\relax
  \def\url#1{\texttt{#1}}\fi
\expandafter\ifx\csname urlprefix\endcsname\relax\def\urlprefix{URL }\fi

\bibitem[{Alvo et~al.(2023)Alvo, Russo and Kanoria}]{alvo2023neural}
\textsc{Alvo, M.}, \textsc{Russo, D.} and \textsc{Kanoria, Y.} (2023).
\newblock Neural inventory control in networks via hindsight differentiable
  policy optimization.
\newblock {\href{https://arxiv.org/abs/2306.11246}{\texttt{arXiv:2306.11246}}}.

\bibitem[{Andaz et~al.(2023)Andaz, Eisenach, Madeka, Torkkola, Jia, Foster and
  Kakade}]{andaz2023learning}
\textsc{Andaz, S.}, \textsc{Eisenach, C.}, \textsc{Madeka, D.},
  \textsc{Torkkola, K.}, \textsc{Jia, R.}, \textsc{Foster, D.} and
  \textsc{Kakade, S.} (2023).
\newblock Learning an inventory control policy with general inventory arrival
  dynamics.
\newblock {\href{https://arxiv.org/abs/2310.17168}{\texttt{arXiv:2310.17168}}}.

\bibitem[{Arrow et~al.(1958)Arrow, Karlin, Scarf et~al.}]{arrow1958studies}
\textsc{Arrow, K.~J.}, \textsc{Karlin, S.}, \textsc{Scarf, H.~E.}
  \textsc{et~al.} (1958).
\newblock \textit{Studies in the mathematical theory of inventory and
  production}.
\newblock Stanford University Press.

\bibitem[{Boyd et~al.(2004)Boyd, Boyd, Vandenberghe and Press}]{boyd2004convex}
\textsc{Boyd, S.}, \textsc{Boyd, S.}, \textsc{Vandenberghe, L.} and
  \textsc{Press, C.~U.} (2004).
\newblock \textit{Convex Optimization}.
\newblock pt. 1, Cambridge University Press.

\bibitem[{Boyd et~al.(2011)Boyd, Parikh, Chu, Peleato and
  Eckstein}]{boyd2011foundations}
\textsc{Boyd, S.~P.}, \textsc{Parikh, N.}, \textsc{Chu, E.}, \textsc{Peleato,
  B.} and \textsc{Eckstein, J.} (2011).
\newblock Distributed optimization and statistical learning via the alternating
  direction method of multipliers.
\newblock \textit{Foundations and Trends in Machine Learning} \textbf{3}
  1--122.

\bibitem[{Bretthauera et~al.(1994)Bretthauera, Shetty, Syam and
  White}]{bretthauer1994model}
\textsc{Bretthauera, K.}, \textsc{Shetty, B.}, \textsc{Syam, S.} and
  \textsc{White, S.} (1994).
\newblock A model for resource constrained production and inventory management.
\newblock \textit{Decision Sciences} \textbf{25} 561--577.

\bibitem[{Camacho and Bordons(2004)}]{camacho2004model}
\textsc{Camacho, E.} and \textsc{Bordons, C.} (2004).
\newblock \textit{Model Predictive Control}.
\newblock Advanced Textbooks in Control and Signal Processing, Springer London.

\bibitem[{Caro and Gallien(2010)}]{caro2010inventory}
\textsc{Caro, F.} and \textsc{Gallien, J.} (2010).
\newblock Inventory management of a fast-fashion retail network.
\newblock \textit{Operations Research} \textbf{58} 257--273.

\bibitem[{Donoho(1992)}]{donoho1992unconditional}
\textsc{Donoho, D.} (1992).
\newblock {Unconditional Bases are Optimal Bases for Data Compression and for
  Statistical Estimation}.
\newblock \textit{Applied and Computational Harmonic Analysis} \textbf{1}
  100--115.

\bibitem[{Eisenach et~al.(2020)Eisenach, Patel and Madeka}]{eisenach2020mqt}
\textsc{Eisenach, C.}, \textsc{Patel, Y.} and \textsc{Madeka, D.} (2020).
\newblock {MQTransformer: Multi-Horizon Forecasts with Context Dependent and
  Feedback-Aware Attention}.
\newblock {\href{https://arxiv.org/abs/2009.14799}{\texttt{arXiv:2009.14799}}}.

\bibitem[{García et~al.(1989)García, Prett and Morari}]{garcia1989mpc}
\textsc{García, C.~E.}, \textsc{Prett, D.~M.} and \textsc{Morari, M.} (1989).
\newblock Model predictive control: Theory and practice - a survey.
\newblock \textit{Automatica} \textbf{25} 335--348.

\bibitem[{Gijsbrechts et~al.(2022)Gijsbrechts, Boute, Van~Mieghem and
  Zhang}]{gijsbrechts2022can}
\textsc{Gijsbrechts, J.}, \textsc{Boute, R.~N.}, \textsc{Van~Mieghem, J.~A.}
  and \textsc{Zhang, D.~J.} (2022).
\newblock Can deep reinforcement learning improve inventory management?
  performance on lost sales, dual-sourcing, and multi-echelon problems.
\newblock \textit{Manufacturing \& Service Operations Management} \textbf{24}
  1349--1368.

\bibitem[{Haar(1910)}]{haar1910zur}
\textsc{Haar, A.} (1910).
\newblock Zur theorie der orthogonalen funktionensysteme.
\newblock \textit{Mathematische Annalen} \textbf{69} 331--371.

\bibitem[{Kwon et~al.(1983)Kwon, Bruckstein and Kailath}]{kwon1983stabilizing}
\textsc{Kwon, W.-H.}, \textsc{Bruckstein, A.} and \textsc{Kailath, T.} (1983).
\newblock Stabilizing state-feedback design via the moving horizon method.
\newblock vol.~37.

\bibitem[{Lo and Topaloglu(2021)}]{lo2021omnichannel}
\textsc{Lo, V.} and \textsc{Topaloglu, H.} (2021).
\newblock Omnichannel assortment optimization under the multinomiallogit model
  with a features tree.
\newblock \textit{Manufacturing \& Service Operations Management} \textbf{24}
  1220--1240.

\bibitem[{Madeka et~al.(2022)Madeka, Torkkola, Eisenach, Luo, Foster and
  Kakade}]{madeka2022deep}
\textsc{Madeka, D.}, \textsc{Torkkola, K.}, \textsc{Eisenach, C.}, \textsc{Luo,
  A.}, \textsc{Foster, D.} and \textsc{Kakade, S.} (2022).
\newblock Deep inventory management.
\newblock {\href{https://arxiv.org/abs/2210.03137}{\texttt{arXiv:2210.03137}}}.

\bibitem[{Maggiar et~al.(2024)Maggiar, Dicker and Mahoney}]{maggiar2024cpp}
\textsc{Maggiar, A.}, \textsc{Dicker, L.} and \textsc{Mahoney, M.~W.} (2024).
\newblock {Consensus Planning with Primal, Dual, and Proximal Agents}.
\newblock {\href{https://arxiv.org/abs/2408.16462}{\texttt{arXiv:2408.16462}}}.

\bibitem[{Maggiar et~al.(2022)Maggiar, Song and
  Muharremoglu}]{maggiar2022multi}
\textsc{Maggiar, A.}, \textsc{Song, I.} and \textsc{Muharremoglu, A.} (2022).
\newblock Multi-echelon inventory management for a non-stationary capacitated
  distribution network.
\newblock Tech. rep., SSRN.
{\href{https://papers.ssrn.com/sol3/Delivery.cfm/SSRN_ID4154780_code1172651.pdf?abstractid=4154780&mirid=1&type=2}{[PDF]}}

\bibitem[{Maloney and Klein(1993)}]{maloney1993constrained}
\textsc{Maloney, B.~M.} and \textsc{Klein, C.~M.} (1993).
\newblock Constrained multi-item inventory systems: An implicit approach.
\newblock \textit{Computers \& Operations Research} \textbf{20} 639--649.

\bibitem[{Mnih et~al.(2016)Mnih, Badia, Mirza, Graves, Lillicrap, Harley,
  Silver and Kavukcuoglu}]{mnih2016asynchronous}
\textsc{Mnih, V.}, \textsc{Badia, A.~P.}, \textsc{Mirza, M.}, \textsc{Graves,
  A.}, \textsc{Lillicrap, T.~P.}, \textsc{Harley, T.}, \textsc{Silver, D.} and
  \textsc{Kavukcuoglu, K.} (2016).
\newblock Asynchronous methods for deep reinforcement learning.
\newblock {\href{https://arxiv.org/abs/1602.01783}{\texttt{arXiv:1602.01783}}}.

\bibitem[{Mnih et~al.(2013)Mnih, Kavukcuoglu, Silver, Graves, Antonoglou,
  Wierstra and Riedmiller}]{mnih2013playing}
\textsc{Mnih, V.}, \textsc{Kavukcuoglu, K.}, \textsc{Silver, D.},
  \textsc{Graves, A.}, \textsc{Antonoglou, I.}, \textsc{Wierstra, D.} and
  \textsc{Riedmiller, M.} (2013).
\newblock Playing atari with deep reinforcement learning.
\newblock {\href{https://arxiv.org/abs/1312.5602}{\texttt{arXiv:1312.5602}}}.

\bibitem[{Mousa et~al.(2023)Mousa, van~de Berg, Kotecha, del Rio-Chanona and
  Mowbray}]{mousa2023analysis}
\textsc{Mousa, M.}, \textsc{van~de Berg, D.}, \textsc{Kotecha, N.}, \textsc{del
  Rio-Chanona, E.~A.} and \textsc{Mowbray, M.} (2023).
\newblock An analysis of multi-agent reinforcement learning for decentralized
  inventory control systems.
\newblock {\href{https://arxiv.org/abs/2307.11432}{\texttt{arXiv:2307.11432}}}.

\bibitem[{Nahmias(1979)}]{nahmias1979simple}
\textsc{Nahmias, S.} (1979).
\newblock Simple approximations for a variety of dynamic leadtime lost-sales
  inventory models.
\newblock \textit{Operations Research} \textbf{27} 904--924.

\bibitem[{Parmas et~al.(2023)Parmas, Seno and Aoki}]{parmas2023model}
\textsc{Parmas, P.}, \textsc{Seno, T.} and \textsc{Aoki, Y.} (2023).
\newblock Model-based reinforcement learning with scalable composite policy
  gradient estimators.
\newblock In \textit{ICML}.

\bibitem[{Porteus(2002)}]{porteus2002foundations}
\textsc{Porteus, E.~L.} (2002).
\newblock \textit{Foundations of stochastic inventory theory}.
\newblock Stanford University Press.

\bibitem[{Qi et~al.(2023)Qi, Shi, Qi, Ma, Yuan, Wu and Shen}]{qi2023practical}
\textsc{Qi, M.}, \textsc{Shi, Y.}, \textsc{Qi, Y.}, \textsc{Ma, C.},
  \textsc{Yuan, R.}, \textsc{Wu, D.} and \textsc{Shen, Z.-J.} (2023).
\newblock A practical end-to-end inventory management model with deep learning.
\newblock \textit{Management Science} \textbf{69} 759--773.

\bibitem[{Rosenblatt(1981)}]{rosenblatt1981multi}
\textsc{Rosenblatt, M.} (1981).
\newblock Multi-item inventory system with budgetary constraint: a comparison
  between the lagrangian and the fixed cycle approach.
\newblock \textit{International Journal of Production Research} \textbf{19}
  331--339.

\bibitem[{Rosenblatt and Rothblum(1990)}]{rosenblatt1990single}
\textsc{Rosenblatt, M.~J.} and \textsc{Rothblum, U.~G.} (1990).
\newblock On the single resource capacity problem for multi-item inventory
  systems.
\newblock \textit{Operations Research} \textbf{38} 686--693.

\bibitem[{Ross et~al.(2011)Ross, Gordon and Bagnell}]{ross2011reduction}
\textsc{Ross, S.}, \textsc{Gordon, G.} and \textsc{Bagnell, D.} (2011).
\newblock A reduction of imitation learning and structured prediction to
  no-regret online learning.
\newblock In \textit{Proceedings of the Fourteenth International Conference on
  Artificial Intelligence and Statistics}.

\bibitem[{Scarf(1959)}]{scarf1959opt}
\textsc{Scarf, H.} (1959).
\newblock {The optimality of (s,s) policies in the dynamic inventory problem}.
\newblock Tech. rep., Stanford University.

\bibitem[{Schulman et~al.(2017)Schulman, Wolski, Dhariwal, Radford and
  Klimov}]{schulman2017proximal}
\textsc{Schulman, J.}, \textsc{Wolski, F.}, \textsc{Dhariwal, P.},
  \textsc{Radford, A.} and \textsc{Klimov, O.} (2017).
\newblock Proximal policy optimization algorithms.
\newblock {\href{https://arxiv.org/abs/1707.06347}{\texttt{arXiv:1707.06347}}}.

\bibitem[{Silver et~al.(2016)Silver, Huang, Maddison, Guez, Sifre, Van
  Den~Driessche, Schrittwieser, Antonoglou, Panneershelvam, Lanctot
  et~al.}]{silver2016mastering}
\textsc{Silver, D.}, \textsc{Huang, A.}, \textsc{Maddison, C.~J.},
  \textsc{Guez, A.}, \textsc{Sifre, L.}, \textsc{Van Den~Driessche, G.},
  \textsc{Schrittwieser, J.}, \textsc{Antonoglou, I.}, \textsc{Panneershelvam,
  V.}, \textsc{Lanctot, M.} \textsc{et~al.} (2016).
\newblock Mastering the game of go with deep neural networks and tree search.
\newblock \textit{Nature} \textbf{529} 484--489.

\bibitem[{Sinclair et~al.(2023)Sinclair, Vieira~Frujeri, Cheng, Marshall,
  Barbalho, Li, Neville, Menache and Swaminathan}]{sinclair2023hindsight}
\textsc{Sinclair, S.~R.}, \textsc{Vieira~Frujeri, F.}, \textsc{Cheng, C.-A.},
  \textsc{Marshall, L.}, \textsc{Barbalho, H. D.~O.}, \textsc{Li, J.},
  \textsc{Neville, J.}, \textsc{Menache, I.} and \textsc{Swaminathan, A.}
  (2023).
\newblock Hindsight learning for {MDP}s with exogenous inputs.
\newblock In \textit{Proceedings of the 40th International Conference on
  Machine Learning}, vol. 202 of \textit{Proceedings of Machine Learning
  Research}. PMLR.

\bibitem[{Sutton and Barto(2020)}]{sutton2020reinforcement}
\textsc{Sutton, R.~S.} and \textsc{Barto, A.~G.} (2020).
\newblock \textit{Reinforcement Learning: An iIntroduction}.
\newblock MIT press.

\bibitem[{Szepesvári(2010)}]{szepesvari2010algorithms}
\textsc{Szepesvári, C.} (2010).
\newblock \textit{Algorithms for Reinforcement Learning}.
\newblock Synthesis Lectures on Artificial Intelligence and Machine Learning,
  Morgan \& Claypool Publishers.

\bibitem[{Thomas(2023)}]{thomas2023towards}
\textsc{Thomas, J.~D.} (2023).
\newblock Towards cooperative marl in industrial domains.

\bibitem[{Tripuraneni et~al.(2021)Tripuraneni, Madeka, Foster, Perrault-Joncas
  and Jordan}]{tripuraneni2021meta}
\textsc{Tripuraneni, N.}, \textsc{Madeka, D.}, \textsc{Foster, D.},
  \textsc{Perrault-Joncas, D.} and \textsc{Jordan, M.~I.} (2021).
\newblock Meta-analysis of randomized experiments with applications to
  heavy-tailed response data.
\newblock {\href{https://arxiv.org/abs/2112.07602}{\texttt{arXiv:2112.07602}}}.

\bibitem[{van~den Oord et~al.(2016)van~den Oord, Dieleman, Zen, Simonyan,
  Vinyals, Graves, Kalchbrenner, Senior and Kavukcuoglu}]{oord2016wavenet}
\textsc{van~den Oord, A.}, \textsc{Dieleman, S.}, \textsc{Zen, H.},
  \textsc{Simonyan, K.}, \textsc{Vinyals, O.}, \textsc{Graves, A.},
  \textsc{Kalchbrenner, N.}, \textsc{Senior, A.} and \textsc{Kavukcuoglu, K.}
  (2016).
\newblock Wavenet: A generative model for raw audio.
\newblock {\href{https://arxiv.org/abs/1609.03499}{\texttt{arXiv:1609.03499}}}.

\bibitem[{Ventura and Klein(1988)}]{ventura1988note}
\textsc{Ventura, J.~A.} and \textsc{Klein, C.~M.} (1988).
\newblock A note on multi-item inventory systems with limited capacity.
\newblock \textit{Operations Research Letters} \textbf{7} 71--75.

\bibitem[{Wen et~al.(2017)Wen, Torkkola, Narayanaswamy and
  Madeka}]{wen2017mqcnn}
\textsc{Wen, R.}, \textsc{Torkkola, K.}, \textsc{Narayanaswamy, B.} and
  \textsc{Madeka, D.} (2017).
\newblock {A multi-horizon quantile recurrent forecaster}.
\newblock In \textit{NIPS Time Series Workshop}.

\bibitem[{Xie et~al.(2024)Xie, Ma and Xin}]{xie2023vcinv}
\textsc{Xie, Y.}, \textsc{Ma, W.} and \textsc{Xin, L.} (2024).
\newblock Vc theory for inventory policies.
\newblock {\href{https://arxiv.org/abs/2404.11509}{\texttt{arXiv:2404.11509}}}.

\bibitem[{Zhao et~al.(2023)Zhao, Tang and Yao}]{zhao2023policy}
\textsc{Zhao, H.}, \textsc{Tang, W.} and \textsc{Yao, D.~D.} (2023).
\newblock Policy optimization for continuous reinforcement learning.
\newblock {\href{https://arxiv.org/abs/2305.18901}{\texttt{arXiv:2305.18901}}}.

\bibitem[{Ziegler(1982)}]{ziegler1982solving}
\textsc{Ziegler, H.} (1982).
\newblock Solving certain singly constrained convex optimization problems in
  production planning.
\newblock \textit{Operations Research Letters} \textbf{1} 246--252.

\bibitem[{Zipkin(2008)}]{zipkin2008old}
\textsc{Zipkin, P.} (2008).
\newblock Old and new methods for lost-sales inventory systems.
\newblock \textit{Operations research} \textbf{56} 1256--1263.

\end{thebibliography}

\clearpage
\appendix

\section{Additional Experimental Results}
\label{sec:experiments-more}
\Cref{fig:mpc-vs-nc-full1} and \Cref{fig:mpc-vs-nc-full2} shows both the final costs for the base stock policy with MPC coordinator and the RL policy with neural coordinator for a randomly selected constraint trajectory. As can be seen, the costs produced by the neural coordinator appear to be martingale, while the MPC produced costs do not. MPC cost trajectories only start 5 weeks out as that was the longest planning horizon of any product in the evaluation dataset. We see similar trends in the historic capacity costs -- see \Cref{sec:costs-acc-evo}.

\begin{figure}[h]
	\centering
	\amazon{\includegraphics[width=\textwidth]{plots/css_mpc_vs_nc_full1.pdf}}{\includegraphics[width=\textwidth]{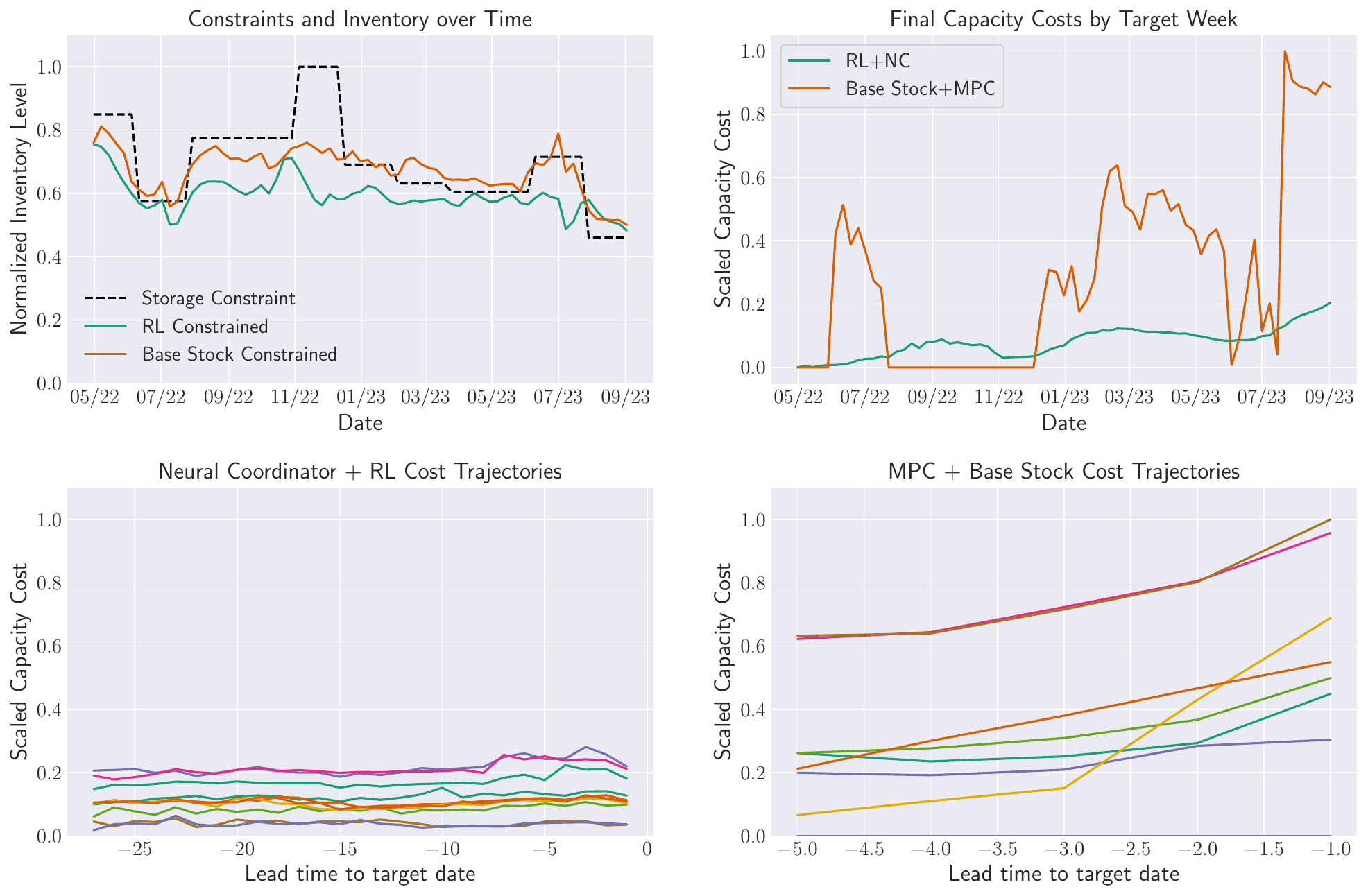}}
	\caption{Comparison of cost and capacity trajectories under base stock with MPC and RL with the neural coordinator.}
	\label{fig:mpc-vs-nc-full1}
\end{figure}

\begin{figure}[h]
	\centering
	\amazon{\includegraphics[width=\textwidth]{plots/css_mpc_vs_nc_full2.pdf}}{\includegraphics[width=\textwidth]{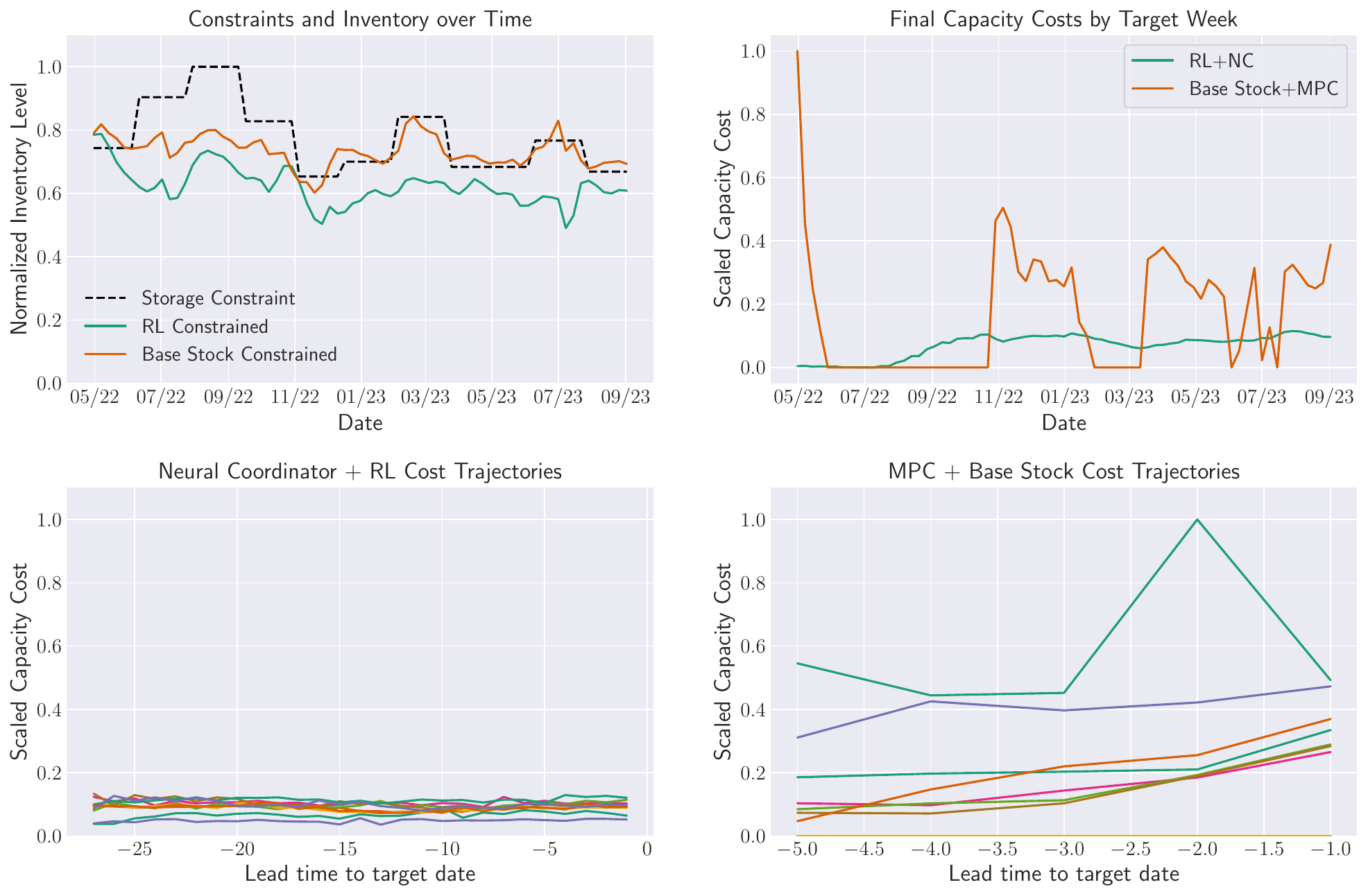}}
	\caption{Comparison of cost and capacity trajectories under base stock with MPC and RL with the neural coordinator.}
	\label{fig:mpc-vs-nc-full2}
\end{figure}

\clearpage
\section{Data and Features}
\label{sec:features}

\subsubsection*{Data}
In this section we first describe some salient properties of our data (all values are normalized and axes removed) that impact the required featurization for the neural coordinator. Specifically, there are two properties which are relevant for the neural coordinator:
\begin{itemize}
	\item our data is heavy tailed in nature  \citep{tripuraneni2021meta}
	\item and the exogenous variables (e.g. prices and costs) are not independent.
\end{itemize}

Denote the mean demand for product $i$ as $\bar{d}^i := \frac{1}{T}\sum_{t=1}^T d_t^i$, and similarly denote the mean costs and prices as $\bar{c}^i$ and $\bar{p}^i$, respectively. \Cref{fig:densities} shows a joint kernel density estimate of reward per sale and demand, along with the joint density estimate of the price and cost components in the reward function. Because the neural coordinator is solving a problem that depends on the joint distribution at the ASIN-level and \Cref{fig:densities} shows there is correlation across the relevant exogenous variables, the neural coordinator needs a carefully selected set of features so that it can accurately forecast the costs needed to constrain.  Finally, the right-most plot in \Cref{fig:densities} shows the heavy-tailed nature of the reward distribution.

\begin{figure}[h]
	\centering
	\includegraphics[width=0.3\textwidth]{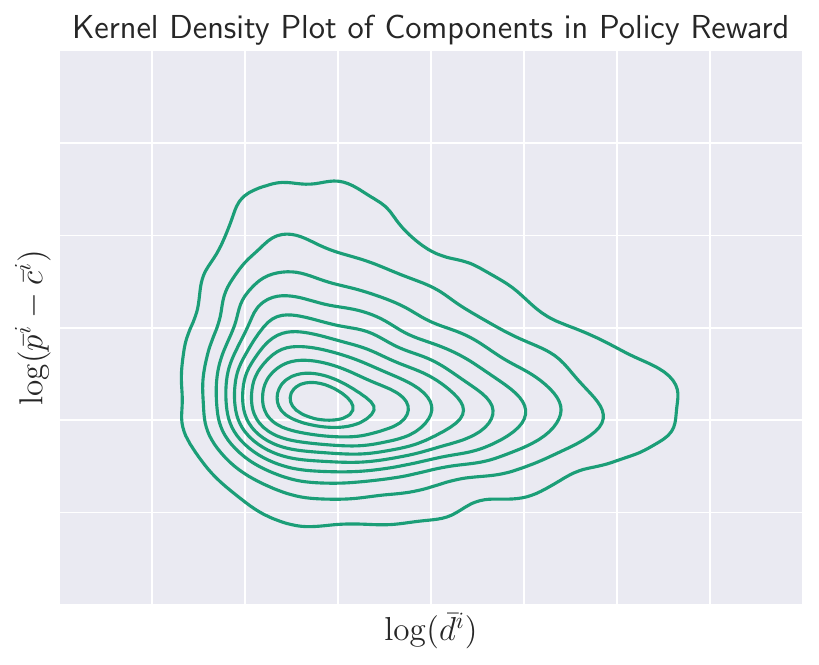} ~~ \includegraphics[width=0.3\textwidth]{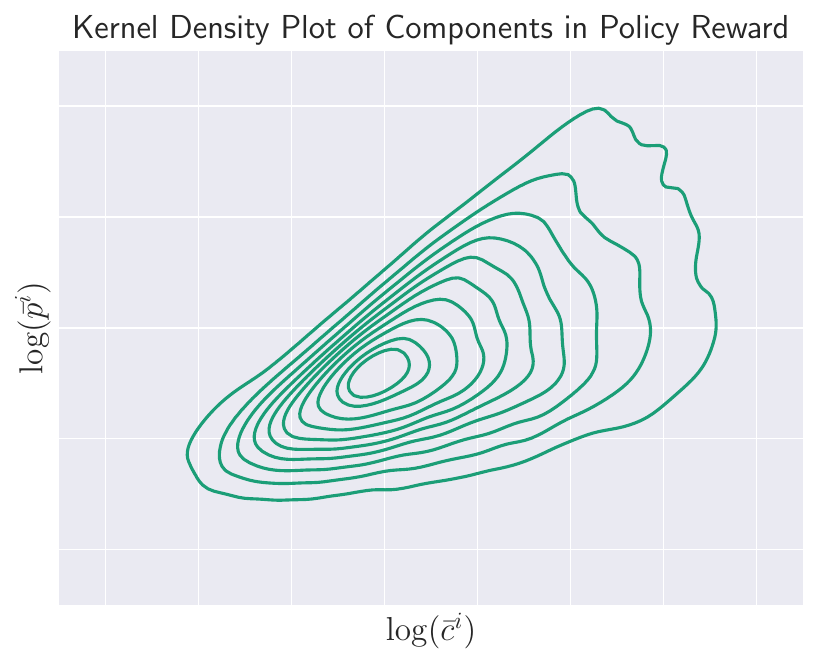} ~~ \includegraphics[width=0.3\textwidth]{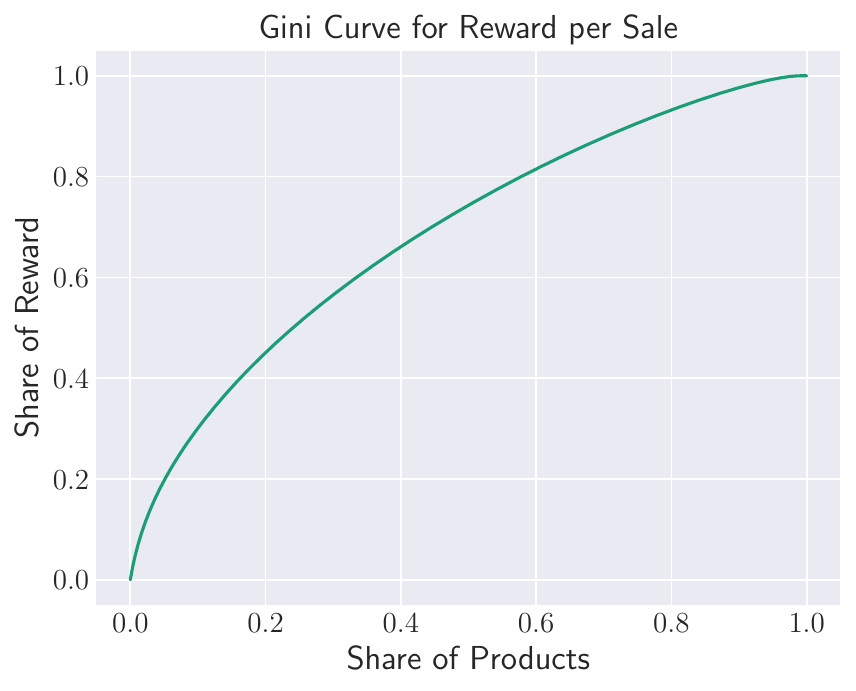} 
	\caption{ Density plots showing components of the reward function and a gini plot showing the heavy-tailed nature of the reward-per-sale distribution. } 
	\label{fig:densities}
\end{figure}


\subsubsection*{Featurization for Buying Policy}
In terms of features, we mainly follow \citet{madeka2022deep} in terms of features provided to the buying policy:
\begin{enumerate}
	\item The current inventory level $I^i_{(t-1)}$
	\item Previous actions $a^i_u$ that have been taken $\forall u < t$
	\item Time series features
		  \begin{enumerate}
			  \item Historical availability corrected demand
			  \item Distance to public holidays
			  \item Historical website glance views data
		  \end{enumerate}
	\item Static product features
		  \begin{enumerate}
			  \item Product group
			  \item Text-based features from the product description
			  \item Brand
			  \item Volume
		  \end{enumerate}
	\item Economics of the product - (price, cost etc.)
	\item Capacity costs -- past costs and forecasted future costs
\end{enumerate}

\subsubsection*{Featurization for Neural Coordinator}

The neural coordinator takes the following aggregate / population level features:
\begin{enumerate}
	\item Aggregate features for all current and previous times:
	\begin{enumerate}
		\item Order quantities: $\sum_{i\in\cA} a^i_{t}$
		\item Inventory: $\sum_{i\in\cA} D^i_{(t-1)}$
		\item Availability corrected demand $\tilde{D}_{t} := \sum_{i\in\cA} \hat{D}^i_{t}$
		\item Inbound: $\sum_{i\in\cA} J^i_{t}$
		\item All the above, but weighted by inbound ($u^i$) and storage volumes ($v^i$)
	\end{enumerate}	
	\item Forecasted quantities for next $L$ weeks.
	\begin{enumerate}
		\item Mean demand at lead time $l$: $\hat{D}_{t,l} := \sum_{i\in A} \hat{D}^i_{t,l}$
		\item Inventory after expected drain at lead time $l$: $\sum_{i\in A} \left[I^i_t - \sum_{s=t}^{t+l} \hat{D}^i_{t,l}\right]$
		\item All the above, but weighted by inbound ($u^i$) and storage volumes ($v^i$)
	\end{enumerate}
	\item Other time series features
		  \begin{enumerate}
			  \item Distance to public holidays
			  \item Mean economics of products - (price, cost etc.), weighted by demand and volumes
		  \end{enumerate}
	\item Capacity costs (past costs and forecasted future costs)
\end{enumerate}

\clearpage
\section{Algorithm Details}
\label{sec:alg-details}
\Cref{alg:dbp} describes at a high-level the procedure for learning a buying policy that solves \eqref{eqn:lagrangian-idp}.

\begin{algorithm}[H]
	\caption{Training Algorithm for Lagrangian IDP}
	\label{alg:dbp}
	\begin{algorithmic}
		\State \textbf{Input: } $\cD$ (a set of products), batch size $M$, initial policy parameters $\btheta_0$, $\cG$ (a set of capacity curves sampled from $\PP^G$), quadratic penalty $\alpha$
		\State $b \gets 1$
		\State \texttt{// Initialize dual costs to 0}
		\For{$j \in 1,\dots,|\cG|$}
			\State $\blambda^j_1,\dots, \blambda^j_T \gets 0,\dots,0$
		\EndFor
		\State \texttt{// Run primal-dual procedure until converged}
		\While{not converged}
			\State Sample mini-batch of products $\cD_M$
			\State Sample $j \sim U(|\cG|)$
			\State $G \gets G^j$
			\State $H_1^\lambda,\dots,H_T^\lambda \gets \texttt{CostHistoryFromTrajectory}(\blambda^j_1,\dots,\blambda^j_T)$.
			\State $\tilde{J}_1,\dots \tilde{J}_T \gets 0$ \texttt{// Total inbound}
			\State $\tilde{I}_1,\dots,\tilde{I}_T \gets 0$ \texttt{// Total onhand}
			\State $F_b \gets 0$ \texttt{// Total penalized reward}
			\For{$i \in \cD_M$}
				\State $F_{i}\gets 0$, $I_0^i \gets k^i$
				\For{$t = 0,\dots,T$}
					\State $a^i_{t} = \pi^i_{\theta_b, t}(H_{t})$
					\State $J^i_{t+1} \gets \cT^J_t(H_t, \theta_b)$
					\State $I^i_{t+1} \gets \cT^I_t(H_t, \theta_b)$
					\State $\tilde{J}_{t+1} \gets \tilde{J}_{t+1} + u^iJ^i_{t+1} $
					\State $\tilde{I}_{t+1} \gets \tilde{I}_{t+1} + w^iI^i_{t+1} $
					\State $F_{i}\gets F_{i} + \gamma^t R^{\lambda}(H_t, \theta_b)$
				\EndFor
				\State $F_b \gets F_b + F_{i}$
			\EndFor
			\State $\btheta_{b} \gets \btheta_{b-1} + \alpha \nabla_{\btheta} F_b$
			\State $\blambda^j_1,\dots,\blambda^j_T \gets \texttt{DualUpdate}(G,\blambda^j_1,\dots,\blambda^j_T,\tilde{J}_{1},\dots,\tilde{J}_{T},\tilde{I}_{1},\dots,\tilde{I}_{T},\alpha)$

			\State $b \gets b+1$
		\EndWhile
	\end{algorithmic}
\end{algorithm}
Note that, for this algorithm - at time $t$ the policy only uses the values of the exogenous variables upto time $t$. We can also augment the Lagrangian objective by adding some constant $\alpha$ times the constraint violation
\begin{equation}
	\label{eqn:augmented-lagrangian}
	W_t(\theta, H_t, G) := \Biggl(  \Biggl(\tilde{I}_t - K_t^1\Biggr)_+ \Biggr)^2 + \Biggl(  \Biggl(\tilde{J}_t - K_t^2\Biggr)_+ \Biggr)^2 ,
\end{equation}
for a given constraint sequence $G$.

\section{Learning a Neural Coordinator via Direct-Backprop}
\label{sec:neural-coordinator}
In this section we describe the IDP solved by the coordinator. First, assume a fixed buying policy $\theta$.
It is important to emphasize that the choice of {\bf training objective} for the coordinator is not important as we have a valid backtest (as discussed in \Cref{sec:learn}) to evaluate performance. Below we describe the ways in which the coordinator's IDP deviates from the buying agent's IDP.

\subsubsection*{Global Constraints} As mentioned in \Cref{sec:idp}, we now assume that the global constraint process $G$ is sampled from a distribution $\PP^G$ and is known to the coordinator agent at time $t=0$.\footnote{Other information models, such as revealing only the next $L$ capacity constraints at each time $t$, would be compatible with our approach, but we consider $G$ as being fully known to the coordinator at time 0 for ease of exposition.}

\subsubsection*{Control Processes} The coordinator makes a decision at each time $t$ as to what prices to set for the current time period as well as a forecast of prices for the next $L$ periods. Specifically they are determined by a policy parameterized by $\omega \in \Omega$:
\[
	(\blambda_t, \hat{\blambda}^L_t) = \phi_{\omega,t}(H_t, G).
\]
We characterize the set of these policies as $\Phi = \{\phi_{\omega,t}| \omega \in \Omega, t \in [0,T]\}$. As in \Cref{sec:idpconstruction-lagrangian}, the product level decisions are determined by a policy parameterized by $\theta$:
\[
	a_t^{i} = \pi_{\theta,t}^{i}(H_t,H^\lambda_t).
\]

\subsubsection*{Optimization Problem} Denote the total forecast error for time $t$ as
\[
	L(\blambda_t,H^\lambda_{t}) := \sum_{s=1}^L ||\blambda_t - (\hat{\blambda}^L_{t-s})_{L-s} ||_2^2,
\]
which is the MSE of all past forecasts for the current cost. The coordinator solves the problem given by \eqref{eqn:obj-leon} for some constants $\alpha$ and $\kappa$:
\begin{align*}
	\max_{\omega}&  ~\EE^{\PP,\PP^G}\Biggl[ \sum_{t=0}^T[\alpha (\tilde{I}_t - K_t^1)_+^2 + \alpha (\tilde{J}_t - K_t^2)_+^2 + \kappa ||\blambda_t||_1 + L(\blambda_t, H^\lambda_{t})] \Biggr]    \numberthis \label{eqn:obj-leon}      \\
	\text{subject to: }~~     &    \\
	I^i_0                   & = k^i  \\
	(\blambda_t, \hat{\blambda}^L_t) &= \phi_{\omega,t}(H_t, G) \\
	a^i_t                   & = \pi_{\theta,t}^i(H_t, H^\lambda_t) \\
I^i_{t_-} &= I^i_{t-1} + \sum_{0 \leq k < t} a^i_k \bbone_{v^i_k = (t-k)} \\
I^i_{t} & = \min(I^i_{t_-}- D^i_t, 0) \\
\blambda_t &\geq 0.
\end{align*}
Intuitively, this objective finds the least-cost solution that properly adheres to the capacity constraints.

\subsubsection*{Training Algorithm} We apply direct backprop directly to the optimization problem described above, as everything is fully differentiable. At each pass over the data, we randomly sample a new path $G \sim \PP^G$.

\clearpage
\amazon{
	\input{internal_details/}
}{
\section{Production MPC Capacity Control Costs}
\label{sec:costs-acc-evo}
Costs produced by MPC system are notoriously volatile. \Cref{fig:sort-costs} shows the evolution of inbound costs as the target horizon approaches and \Cref{fig:sortld-costs} shows the evolution of storage costs produced by the MPC system. As can be seen in the images, these costs do not appear to be martingale (which is a property one usually expects costs to satisfy).

\begin{figure}[h]
	\centering
	\includegraphics[width=0.85\textwidth]{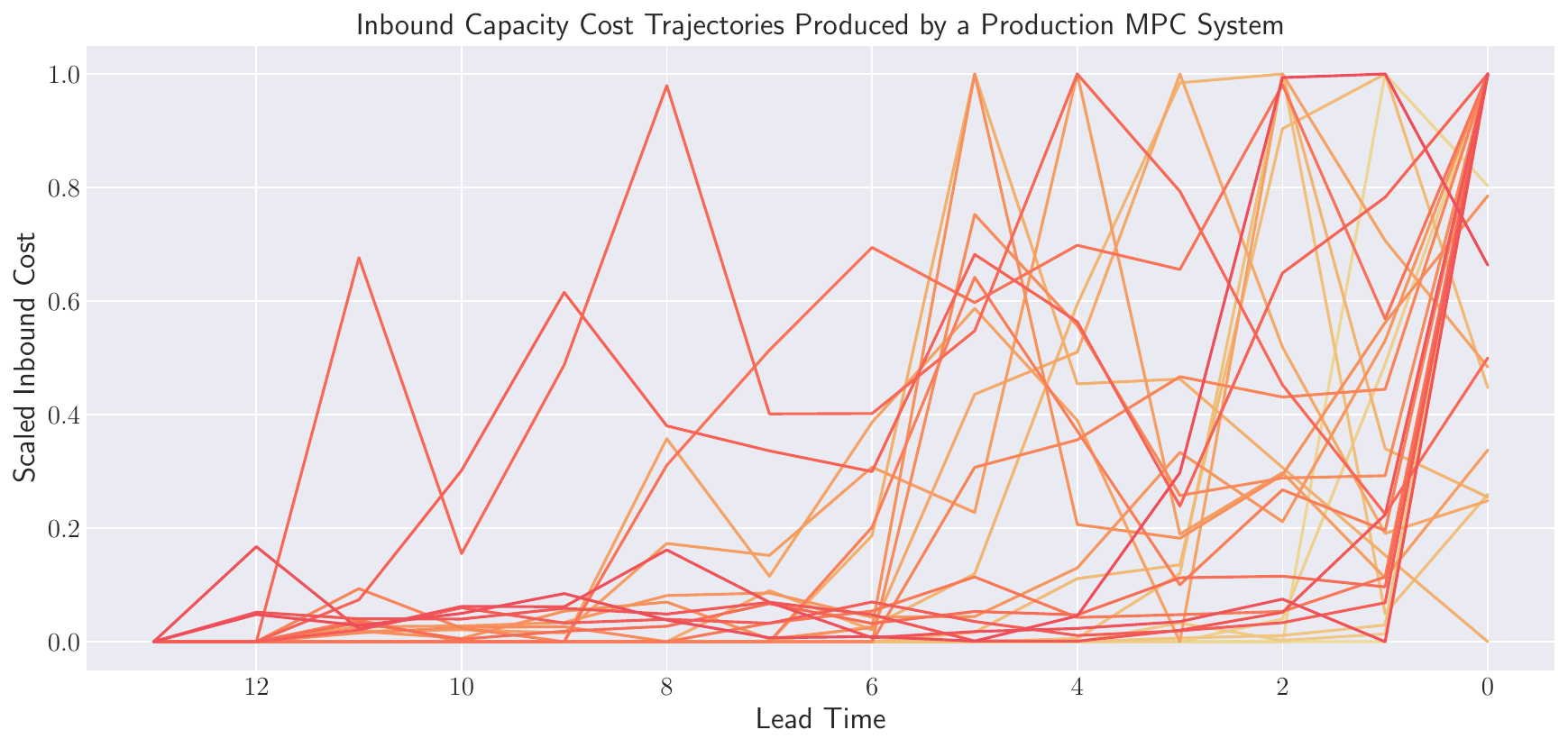}
	\caption{Evolution of inbound costs from a production system at a large e-retailer, scaled by maximum cost along each trajectory.}
	\label{fig:sort-costs}
\end{figure}

\begin{figure}[h]
	\centering
	\includegraphics[width=0.85\textwidth]{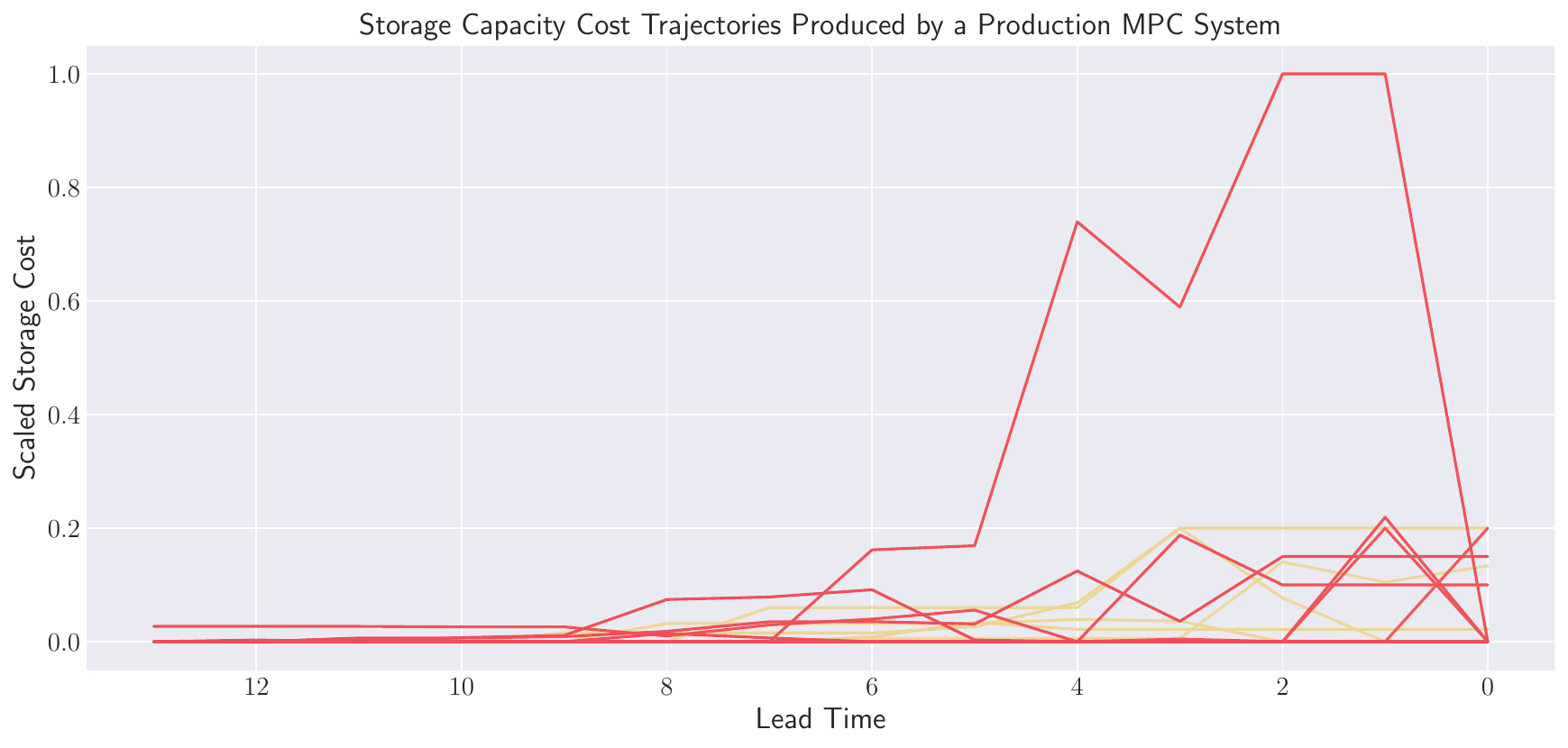}
	\caption{Evolution of storage costs from a production system at a large e-retailer.}
	\label{fig:sortld-costs}
\end{figure}
}

\end{document}